\documentclass[11pt]{article}

\usepackage{amsthm,amsmath,natbib}
\usepackage{parskip}
\usepackage{amsfonts} 
\usepackage{enumerate}
\usepackage{graphicx}
\usepackage{rotating}
\usepackage{natbib}
\newtheorem{proposition}{Proposition}[section]
\newtheorem{lemma}{Lemma}[section]

\newcommand{\R}{\mathbb{R}}
\newcommand{\E}{\mathbb{E}}

\newcommand{\argmin}{\mathop{\arg \min}\limits}
\newcommand{\argmax}{\mathop{\arg \max}\limits}

\newcommand{\calN}{\mathcal{N}}

\begin{document}

\title{Missing values: sparse inverse covariance estimation and an
  extension to sparse regression}

% \author{Nicolas St\"adler\\
% \small Seminar f\"ur Statistik, ETH Z\"urich\\[-0.8ex]
% \small \small CH-8092 Z\"urich, Switzerland.\\
% \small \texttt{staedler@stat.math.ethz.ch}\\
% \and
% Peter B\"uhlmann\\
% \small Seminar f\"ur Statistik, ETH Z\"urich\\[-0.8ex]
% \small \small CH-8092 Z\"urich, Switzerland.\\
% \small \texttt{buhlmann@stat.math.ethz.ch}\\
% }
\author{Nicolas St\"adler and Peter B\"uhlmann\\
\small Seminar for Statistics, ETH Zurich\\[-0.8ex]
\small \small CH-8092 Zurich, Switzerland.\\
\small \texttt{staedler@stat.math.ethz.ch, buhlmann@stat.math.ethz.ch}}

\date{} 

%generates the title
\maketitle

\begin{abstract}
We propose an $\ell_{1}$-regularized likelihood meth\-od for estimating the
inverse covariance matrix in the high-dimensional multivariate normal model
in presence of missing data. Our method is based on the assumption that the
data are missing at random (MAR) which entails also the completely missing
at random case. The implementation
of the method is non-trivial as the observed negative log-likelihood generally is a
complicated and non-convex function. We propose an efficient EM algorithm
for optimization with provable numerical convergence
properties. Furthermore, we extend the methodology to handle missing values
in a sparse regression context. We demonstrate both methods on simulated and real data.\vspace{0.5cm}\\
{\bf Keywords} {Gaussian graphical model, Lasso, Missing data, EM algorithm, Two-stage likelihood}
\vspace{0.5cm}\\
{\bf This is the author’s version of the work (published in Statistics and Computing, 2012, Volume 22, 219-235). The final publication is available at www.springerlink.com.} 
\end{abstract}

\section{Introduction} 
The most common probability model for continuous multivariate data is the
multivariate normal distribution. Many standard methods for analyzing multivariate data, including factor
analysis, principal components and discriminant analysis, are directly
based on the sample mean and covariance matrix of the data.

Another important application are Gaussian graphical models where
conditional dependencies among the variables are entailed in the inverse of
the covariance matrix \citep{lauritzen96graphical}. In particular, the
inverse covariance matrix and its estimate should be sparse having some
entries equaling zero since these encode conditional independencies.
In the context of high-di\-men\-sional data where the number of variables $p$
is much larger than sample size $n$, \cite{meinshausen04consistent}
estimate a sparse Gaussian model by pursuing many $\ell_{1}$-penalized
regressions for every node in the graph and they prove that the procedure
can asymptotically recover the true graph. Later, other authors proposed
algorithms for the exact optimization of the
$\ell_{1}$-penalized log-likelihood (\cite{yuan05model},
\cite{friedman2007sic}, \cite{banerjee2008mst} and
\cite{rothman2008spi}). It has been shown in \cite{wainwright} that such an approach is also
able to recover asymptotically the true graph, but \cite{meinshausen05}
points out that rather restrictive conditions on the true covariance matrix
are necessary. All these approaches and
theoretical analyses have so far been developed for the case where all data
is observed. 
%\cite{banerjee2008mst} use interior point optimization methods. \cite{friedman2007sic}
%propose a simple and fast algorithm (GLasso). Their algorithm cycles through
%the variables, fitting a modified lasso regression to each variable in
%turn. Individual lasso problems are solved by coordinate descent.

However, datasets often suffer from missing values \citep{LittleRubin}. %\citep{Schneider01analysisof} and \citep{LittleRubin}. 
Besides many ad-hoc ap\-proach\-es to the missing-value problem, there is a systematic
approach based on likelihoods which is very popular nowadays (\cite{LittleRubin},
\cite{Schafer}). But even estimation of mean values and covariance matrices becomes
difficult when the data is incomplete and no explicit maximization of the
likelihood is possible. A solution addressing this problem is given by the
EM algorithm for solving missing-data problems based on likelihoods. 

In this article we are interested in estimating the (inverse) covariance
matrix and the mean vector in the high-dimensional multivariate normal
model in presence of missing data, and this in turn allows for imputation. We present a new algorithm for maximizing the $\ell_{1}$-penalized
observed log-likelihood. The proposed method can be used to estimate sparse undirected
graphical models or/and regularized covariance matrices for
high-dimensional data where $p\gg n$. Furthermore, once having a regularized
covariance estimation for the incomplete data at hand, we show how to
do $\ell_1$-penalized regression, when there is an additional response
variable which is regressed on the incomplete data. 

\section{$\ell_{1}$-regularized inverse covariance estimation with missing data}

\subsection{GLasso}\label{subsec:glasso}
Let  $(X^{(1)},\ldots,X^{(p)})$ be Gaussian distributed with
mean $\mu$ and covariance $\Sigma$, i.e., $\calN(\mu,\Sigma)$. We wish to estimate the concentration
matrix $K=\Sigma^{-1}$. Given a complete random sample
$\mathbf{x}=(x_{1},\ldots,x_{n})^T$, \cite{yuan05model} propose to minimize the negative $\ell_{1}$-penalized log-likelihood 
% \begin{align}\label{eq:plikcomplete}
% \!-\ell(\mu,K;\mathbf{x})\!+\!\lambda\|K\|_{1}\!=&\!-\!\frac{n}{2}\!\log|K|\!+\!\frac{1}{2}\!\sum_{i=1}^{n}\!(x_{i}\!-\!\mu)^T\!K\!(x_{i}\!-\!\mu)\nonumber\\
% &+\lambda\|K\|_{1},
% \end{align}
\begin{align}\label{eq:plikcomplete}
&-\ell(\mu,K;\mathbf{x})+\lambda\|K\|_{1}=-\frac{n}{2}\log|K|+\frac{1}{2}\sum_{i=1}^{n}(x_{i}-\mu)^TK(x_{i}-\mu)+\lambda\|K\|_{1},
\end{align}
over non-negative definite matrices $K$ ($K\succ0$), where $\|K\|_{1}=\sum_{j,j'=1}^{p}|K_{jj'}|$. Here $\lambda>0$ is a tuning parameter.

The minimizer $\hat{K}$ is easily seen to satisfy
\begin{eqnarray}\label{eq:glasso}
\hat{K}&=&\argmin_{K\succ0}\big(-\log|K|+\mathrm{tr}(KS)+\rho\|K\|_{1}\big)
\end{eqnarray}
where $S=\frac{1}{n}\sum_{i=1}^{n}(x_{i}-\bar{x})(x_{i}-\bar{x})^T$ and
$\rho=\frac{2\lambda}{n}$.

\cite{friedman2007sic} propose an elegant and efficient algorithm, called
GLasso, to solve the problem (\ref{eq:glasso}). We briefly review the
derivation of their algorithm while details are given in \cite{friedman2007sic} and \cite{banerjee2008mst}. We will make use of this algorithm in the
M-Step of an EM algorithm in a missing data setup, described in Section~\ref{subsec:emalg}.

Using duality, formula (\ref{eq:glasso}) is seen to be equivalent to the maximization problem
\begin{equation}\label{eq:glasso2}
\hat{\Sigma}=\argmax_{\|\Sigma-S\|_{\infty}\leq\rho}\log \det(\Sigma).
\end{equation}
Problem (\ref{eq:glasso2}) can be solved by a block coordinate descent
optimization over each row and corresponding column of $\Sigma$. Partitioning
$\Sigma$ and $S$
\begin{eqnarray*}
\Sigma=\left(\begin{array}{cc}
\Sigma_{11}&\sigma_{12}\\
\sigma_{12}^{T}&\sigma_{22}\end{array}\right),\qquad
S=\left(\begin{array}{cc}
S_{11}&s_{12}\\
s_{12}^{T}&s_{22}\end{array}\right)
\end{eqnarray*}
the block solution for the last column $\sigma_{12}$ satisfies
\begin{equation}\label{eq:glasso3}
\hat{\sigma}_{12}=\argmin_{y: \|(y-s_{12})\|_{\infty}\leq\rho}y^{T}\Sigma_{11}^{-1}y.
\end{equation}
Using duality it can be seen that solving (\ref{eq:glasso3}) is equivalent
to the Lasso problem 
\begin{equation}\label{eq:glasso4}
\hat{\beta}=\argmin_{\beta}\Big(\|\frac{1}{2}\Sigma_{11}^{1/2}\beta-\Sigma_{11}^{-1/2}s_{12}\|_{2}^{2}+\rho\|\beta\|_{1}\Big)
\end{equation}
where $\hat{\sigma}_{12}$ and $\hat{\beta}$ are linked through
$\hat{\sigma}_{12}=\Sigma_{11}\hat{\beta}/2$. Permuting rows and columns so
that the target column is always the last, a Lasso problem like
(\ref{eq:glasso4}) is solved for each column, updating their estimate of $\Sigma$
after each stage. Fast coordinate descent algorithms for the Lasso
\citep{friedman07fastlasso} make this approach very attractive. Although the algorithm solves for $\Sigma$, the
corresponding estimate of $K$ can be recovered cheaply.

\subsection{MissGLasso}\label{sec:missglasso}
We turn now to the situation where some variables are
missing (i.e., not observed).

As before, we assume $(X^{(1)},\ldots,X^{(p)})\sim\calN(\mu,\Sigma)$ to be p-variate normally distributed with
mean $\mu$ and covariance $\Sigma$. We then write $\mathbf{x}=(\mathbf{x}_{\mathrm{obs}},\mathbf{x}_{\mathrm{mis}})$, where
$\mathbf{x}$ represents a random sample of size $n$,
$\mathbf{x}_{\mathrm{obs}}$ denotes the set of observed values, and $\mathbf{x}_{\mathrm{mis}}$ the missing data. Also,
let
\[ \mathbf{x}_{\mathrm{obs}}=(x_{\mathrm{obs},1},x_{\mathrm{obs},2},\ldots,x_{\mathrm{obs},n}),\]
where $x_{\mathrm{obs},i}$ represents the set of variables observed for case $i$, $i=1,\ldots,n$.

A simple way to estimate the concentration matrix $K$ would be to delete
all the cases which contain missing values and then estimating the
covariance by solving the GLasso problem (\ref{eq:glasso}) using only the
complete cases. However, excluding all cases having at least one missing variable can result in a substantial decrease of the sample size available
for the analysis. When $p$ is large relative to $n$ this problem is even
much more pronounced.  

Another ad-hoc method would impute the missing values by the
corresponding mean and then solving the GLasso problem. Such an approach is
typically inferior to what we present below, see also
Sections~\ref{sec:simulation1} and \ref{sec:isoprenoid}.

Much more promising is to base the inference for $\mu$ and $\Sigma$ (or $K$) in presence of
missing values on the observed log-likelihood:
 \begin{align}\label{eq:observedloglik1}
 \ell(\mu,\Sigma;\mathbf{x}_{\mathrm{obs}})=&\!-\frac{1}{2}\sum_{i=1}^{n}\Big(\log|\Sigma_{\mathrm{obs},i}|+(x_{\mathrm{obs},i}-\mu_{\mathrm{obs},i})^T(\Sigma_{\mathrm{obs},i})^{-1}(x_{\mathrm{obs},i}-\mu_{\mathrm{obs},i})\Big)
 \end{align}
% \begin{equation}\label{eq:observedloglik1}
% \ell(\mu,\Sigma;\mathbf{x}_{\mathrm{obs}})=-\frac{1}{2}\sum_{i=1}^{n}\log|\Sigma_{\mathrm{obs},i}|-\frac{1}{2}\sum_{i=1}^{n}(x_{\mathrm{obs},i}-\mu_{\mathrm{obs},i})^T(\Sigma_{\mathrm{obs},i})^{-1}(x_{\mathrm{obs},i}-\mu_{\mathrm{obs},i})
% \end{equation}
where $\mu_{\mathrm{obs},i}$ and $\Sigma_{\mathrm{obs},i}$ are the mean and covariance matrix
of the observed components of $X$ (i.e., $X_{\mathrm{obs}}$) for observation
$i$. Formally (\ref{eq:observedloglik1}) can be re-written in terms of $K$
\begin{align}\label{eq:observedloglik2}
\ell(\mu,K;\mathbf{x}_{\mathrm{obs}})=&\!-\!\frac{1}{2}\!\sum_{i=1}^{n}\!\Big(\!\!\log\!|(\!K^{-1})_{\mathrm{obs},i}|
\!+\!(x_{\mathrm{obs},i}\!-\!\mu_{\mathrm{obs},i})^T\big((K^{-1})_{\mathrm{obs},i}\big)^{-1}(x_{\mathrm{obs},i}-\mu_{\mathrm{obs},i})\Big).
\end{align}
% \begin{equation}\label{eq:observedloglik2}
% \ell(\mu,K;\mathbf{x}_{\mathrm{obs}})
% =-\frac{1}{2}\sum_{i=1}^{n}\log|(K^{-1})_{\mathrm{obs},i}|
% -\frac{1}{2}\sum_{i=1}^{n}(x_{\mathrm{obs},i}-\mu_{\mathrm{obs},i})^T\big((K^{-1})_{\mathrm{obs},i}\big)^{-1}(x_{\mathrm{obs},i}-\mu_{\mathrm{obs},i}).
% \end{equation}

Inference for $\mu$ and $K$ can be based on the
log-likelihood (\ref{eq:observedloglik2}) if we assume that the underlying
missing data mechanism is \emph{ignorable}. The missing data mechanism is
said to be \emph{ignorable} if the probability that an observation is missing may
depend on $\mathbf{x}_{\mathrm{obs}}$ but not on $\mathbf{x}_{\mathrm{mis}}$ (\emph{Missing at Random}) and if the parameters of the data
model and the parameters of the missingness mechanism are
\emph{distinct}. For a precise definition see \cite{LittleRubin}.

Assuming that $p$ is large relative to $n$, we propose for the unknown parameters $(\mu,K)$ the
estimator: 
\begin{eqnarray}
\hat{\mu},\hat{K}&=&\argmin_{(\mu,K):
  K\succ0}-\ell_{\mathrm{pen}}(\mu,K;\mathbf{x}_{\mathrm{obs}})\label{eq:plik}\\
-\ell_{\mathrm{pen}}(\mu,K;\mathbf{x}_{\mathrm{obs}})&=&-\ell(\mu,K;\mathbf{x}_{\mathrm{obs}})+\lambda\|K\|_{1}\label{eq:plik2}
\end{eqnarray}
where $\ell(\mu,K;\mathbf{x}_{\mathrm{obs}})$ is given in (\ref{eq:observedloglik2}). We call
this estimator the \emph{MissGLasso}.

Despite the concise appearance of (\ref{eq:observedloglik2}), the observed log-likelihood
tends to be a complicated (non-convex) function of the individual $\mu_{j}$ and
$K_{j j'}$, $j ,j'=1,\ldots,p$, for a general missing data
pattern, with possible existence of multiple stationary points
(\cite{murray1977}; \cite{Schafer}). Optimization of (\ref{eq:plik}) is a non-trivial
issue. An efficient algorithm is presented in the next section.

\subsection{Computation}
For the derivation of our algorithm presented in Section
\ref{subsec:emalg} we will state first some facts about the conditional
distribution of the Multivariate Normal (MVN) Model. 
\subsubsection{Conditional distribution of the MVN Model and conditional
mean imputation}
Consider a partition $(X_{1},X_{2})\sim\calN(\mu,\Sigma)$. It is well known that
$X_{2}|X_{1}$ follows a linear regression on $X_{1}$ with mean
$\mu_{2}+\Sigma_{21}\Sigma_{11}^{-1}(X_{1}-\mu_{1})$ and covariance
$\Sigma_{22}-\Sigma_{21}\Sigma_{11}^{-1}\Sigma_{12}$
\citep{lauritzen96graphical}. Thus,
\begin{equation}\label{eq:cond1}
X_{2}|X_{1}\!\sim\calN\big(\mu_{2}\!+\!\Sigma_{21}\Sigma_{11}^{-1}\!(X_{1}-\mu_{1}), \Sigma_{22}\!-\!\Sigma_{21}\Sigma_{11}^{-1}\Sigma_{12}\big).
\end{equation}
Expanding the identity $K\Sigma=I$ gives the following useful expression:
\begin{eqnarray}\label{eq:identity}
\left(\begin{array}{cc}
K_{11}&K_{12}\\
K_{21}&K_{22}\end{array}\right)
\left(\begin{array}{cc}
\Sigma_{11}&\Sigma_{12}\\
\Sigma_{21}&\Sigma_{22}\end{array}\right)&=&
\left(\begin{array}{cc}
I&0\\
0&I\end{array}\right).
\end{eqnarray}
Using (\ref{eq:identity}) we can re-express (\ref{eq:cond1}) in terms of
$K$:
\begin{equation}\label{eq:cond2}
X_{2}|X_{1}\sim\calN\big(\mu_{2}-K_{22}^{-1}K_{21}(X_{1}-\mu_{1}), K_{22}^{-1}\big).
\end{equation}

Formula (\ref{eq:cond2}) will be used later in our developed EM algorithm for estimation of the mean
$\mu$ and the concentration matrix $K$ based on a random sample with
missing values. 

The spirit of this EM algorithm, see Section~\ref{subsec:emalg}, is captured by the following method of imputing missing values
by conditional means
due to \cite{Buck}:
\begin{itemize}
\item[1.] Estimate $(\mu,K)$ by solving the GLasso problem
  (\ref{eq:glasso}) using only the complete cases (delete the rows with
  missing values). This gives estimates $\hat{\mu}$, $\hat{K}$.
\item[2.] Use these estimates to calculate the least squares linear
  regressions of the missing variables on the present variables, case by
  case:
From the above discussion about the multivariate normal distribution, the
missing variables of case $i$, $x_{\mathrm{mis},i}$, given
$x_{\mathrm{obs},i}$ are normally distributed with mean
\begin{align*}
\E[x_{\mathrm{mis},i}|x_{\mathrm{obs},i},\mu,K]=&\mu_{\mathrm{mis}}-(K_{\mathrm{mis},\mathrm{mis}})^{-1}K_{\mathrm{mis},\mathrm{obs}}\left(x_{\mathrm{obs},i}-\mu_{\mathrm{obs}}\right).
\end{align*}
 Therefore an imputation of
  the missing values can be done by
 \[\hat{x}_{\mathrm{mis},i}:=\hat{\mu}_{\mathrm{mis}}-(\hat{K}_{\mathrm{mis},\mathrm{mis}})^{-1}\hat{K}_{\mathrm{mis},\mathrm{obs}}\left(x_{\mathrm{obs},i}-\hat{\mu}_{\mathrm{obs}}\right).\]
Here, $\hat{\mu}_{\mathrm{obs}}$ and $\hat{\mu}_{\mathrm{mis}}$ depend on case $i$. Furthermore, $\hat{K}_{\mathrm{mis},\mathrm{mis}}$ denotes the sub-matrix of $\hat{K}$ with rows and columns corresponding
  to the missing variables for case $i$. Similarly $\hat{K}_{\mathrm{mis},\mathrm{obs}}$ denotes
  the sub-matrix with rows corresponding to the missing variables and
  columns corresponding to the observed variables for case $i$. Note that
  we always notationally suppress the dependence on $i$.
\item[3.] Finally, re-estimate $(\mu,K)$ by solving the GLasso problem on the
  completed data in step 2.
\end{itemize}

\subsubsection{$\ell_{1}$-norm penalized likelihood estimation via the EM algorithm}
\label{subsec:emalg}

A convenient method for optimizing incomplete data problems like (\ref{eq:plik}) is the
EM algorithm (\cite{Dempster}).

To derive the EM algorithm for minimizing (\ref{eq:plik}) we note that the
complete data follows a multivariate normal distribution, which belongs to
the regular exponential family with sufficient statistics 
\[\mathbf{T}_{1}=\mathbf{x}^T1=\left(\sum_{i=1}^{n}x_{i1}, \sum_{i=1}^{n}x_{i2},
  \ldots,\sum_{i=1}^{n}x_{ip}\right) \]
and
\begin{align*}
\mathbf{T}_{2}&=\mathbf{x}^T\mathbf{x}=\left(\begin{array}{cccc}
\sum_{i=1}^{n}x_{i1}^{2}&\sum_{i=1}^{n}x_{i1}x_{i2}&\ldots&\sum_{i=1}^{n}x_{i1}x_{ip}\\
\sum_{i=1}^{n}x_{i2}x_{i1}&\sum_{i=1}^{n}x_{i2}^{2}&\ldots&\sum_{i=1}^{n}x_{i2}x_{ip}\\
\vdots&\vdots&&\vdots\\
\sum_{i=1}^{n}x_{ip}x_{i1}&\sum_{i=1}^{n}x_{ip}x_{i2}&\ldots&\sum_{i=1}^{n}x_{ip}^{2}\\
\end{array}\right).
\end{align*}

The complete penalized negative log-likelihood (\ref{eq:plikcomplete}) can be expressed in term of the sufficient statistics $\mathbf{T}_{1}$ and $\mathbf{T}_{2}$:
  \begin{align}\label{eq:completeplik}
   -\ell(\mu,K;\mathbf{x})+\!\lambda\|K\|_{1}=&\!
   -\frac{n}{2}\!\log\!|K|\!+\frac{n}{2}\mu^TK\mu-\!\mu^TK\mathbf{T}_{1}+\frac{1}{2}\mathrm{tr}(K\mathbf{T}_{2})+\lambda\|K\|_{1}
  \end{align}
which is linear in $\mathbf{T}_{1}$ and $\mathbf{T}_{2}$. The expected complete penalized
log-likelihood is denoted by:
\[\mathop{Q}(\mu,K|\mu',K')=-\E[\ell(\mu,K;\mathbf{x})|\mathbf{x}_{\mathrm{obs}},\mu',K']+\lambda\|K\|_{1}.\]

The EM algorithm works by iterating between the E- and M-Step. Denote the
parameter value at iteration $m$ by $(\mu^{(m)},K^{(m)})$ (m=0,1,2,\ldots),
where $(\mu^{(0)},K^{(0)})$ are the starting values.\\ \\
\textbf{E-Step:} Compute $\mathop{Q}(\mu,K|\mu^{(m)},K^{(m)})$:

As the complete penalized negative log-likelihood in (\ref{eq:completeplik})
is linear in $\mathbf{T}_{1}$ and $\mathbf{T}_{2}$, the E-Step consists of
calculating:
\begin{eqnarray*}
\mathbf{T}_{1}^{(m+1)}=\E[\mathbf{T}_{1}|\mathbf{x}_{\mathrm{obs}},\mu^{(m)},K^{(m)}]\quad
\textrm{and}\quad  
\mathbf{T}_{2}^{(m+1)}=\E[\mathbf{T}_{2}|\mathbf{x}_{\mathrm{obs}},\mu^{(m)},K^{(m)}].
\end{eqnarray*}

This involves computation of the conditional expectation of $x_{ij}$ and
$x_{ij}x_{ij'}$, $i=1,\ldots,n,\; j,j'=1,\ldots,p$. Using
formula (\ref{eq:cond2}) we find
\begin{displaymath}
\E[x_{ij}|x_{\mathrm{obs},i},\mu^{(m)},K^{(m)}]=\left\{\begin{array}{ll}
x_{ij}& \textrm{if $x_{ij}$ observed} \\
c_{j}&%\left[(\mu_{\mathrm{mis}}-K_{\mathrm{mis},\mathrm{mis}}^{-1}K_{\mathrm{mis},\mathrm{obs}}\left(x_{\mathrm{obs},i}-\mu_{\mathrm{obs}}\right)\right]_{j}=c_{j}&
\textrm{if $x_{ij}$ missing}
\end{array}\right.
\end{displaymath}
where $c$ is defined as 
\[c:=\mu^{(m)}_{\mathrm{mis}}-(K^{(m)}_{\mathrm{mis},\mathrm{mis}})^{-1}K^{(m)}_{\mathrm{mis},\mathrm{obs}}\big(x_{\mathrm{obs},i}-\mu^{(m)}_{\mathrm{obs}}\big).\]
Similarly, we compute 
\begin{align*}
&\E[x_{ij}x_{ij'}|x_{\mathrm{obs},i},\mu^{(m)},K^{(m)}]=\left\{\begin{array}{ll}
x_{ij}x_{ij'}&\textrm{if $x_{ij}$ \& $x_{ij'}$ observed,}\\ 
x_{ij}c_{j'}&\textrm{if $x_{ij}$ observed, $x_{ij'}$ missing,}\\
\big(K_{\mathrm{mis},\mathrm{mis}}^{(m)}\big)^{-1}_{jj'}+c_{j}c_{j'}
&\textrm{if $x_{ij}$ \& $x_{ij'}$ missing.}\\
\end{array}\right.
\end{align*}
Here the vector $c$ and the matrix $\big(K_{\mathrm{mis},\mathrm{mis}}^{(m)}\big)^{-1}$ are regarded
as naturally embedded in $\R^{p}$ and $\R^{p\times p}$ respectively, such
that the obvious indexing makes sense.

The E-Step involves inversion of a sparse matrix, namely $K^{(m)}_{\mathrm{mis},\mathrm{mis}}$, for
which we can use sparse linear algebra. Note also that $K^{(m)}_{\mathrm{mis},\mathrm{mis}}$
is positive definite and therefore invertible. Furthermore, considerable savings in
computation are obtained if cases with the same pattern of missing $X$'s
are grouped together.\\ \\
\textbf{M-Step:} Compute the updates $(\mu^{(m+1)},K^{(m+1)})$ as minimizer of $\mathop{Q}(\mu,K|\mu^{(m)},K^{(m)})$: 

It is easily seen from Equation (\ref{eq:completeplik}) that
$\mu^{(m+1)}$ and $K^{(m+1)}$ fulfill the following equations: 
\[\mu^{(m+1)}=\frac{1}{n}\mathbf{T}_{1}^{(m+1)}\]
\[K^{(m+1)}=\argmin_{K\succ0}\Big(-\log|K|+\mathrm{tr}(K\mathbf{S}^{(m+1)})+\frac{2\lambda}{n}\|K\|_{1}\Big)\]
where
$\mathbf{S}^{(m+1)}=\frac{1}{n}\mathbf{T}_{2}^{(m+1)}-\mu^{(m+1)}(\mu^{(m+1)})^T$. Therefore
the M-Step reduces to a GLasso problem of the form (\ref{eq:glasso}), which
can be solved by the algorithm described in
Section~\ref{subsec:glasso}. 

%The described algorithm we call \emph{MissGLasso}.

%\begin{remark}
%E-Step involves inversion of a sparse matrix, namely $K_{\mathrm{mis},\mathrm{mis}}$, for
%which we can use sparse linear algebra. The M-Step reduces to solve a GLasso problem.
%\end{remark} 

\subsubsection{Numerical properties}
A nice property of every EM algorithm is that the objective function is reduced in each iteration,
\[-\ell_{\mathrm{pen}}(\mu^{(m+1)},K^{(m+1)};\mathbf{x}_{\mathrm{obs}})\leq-\ell_{\mathrm{pen}}(\mu^{(m)},K^{(m)};\mathbf{x}_{\mathrm{obs}}).\]
Nevertheless the descent property does not guarantee convergence to a
stationary point. 

A detailed account of the convergence properties of the
EM algorithm in a general setting has been given by \cite{wu}. Under mild regularity
conditions including differentiability and continuity, convergence to
stationary points is proven for the EM algorithm.

For the EM algorithm described in Section~\ref{subsec:emalg} which optimizes
a non-differentiable function we have the following result:
\begin{proposition}\label{prop:convergence}
 Every limit point $(\bar{\mu},\bar{K})$, with $\bar{K}\succ0$, of the sequence $\{(\mu^{(m)},K^{(m)});m =
 0,1,2,\ldots\}$, generated by the EM algorithm, is a stationary point
 of the criterion function in (\ref{eq:plik2}).
\end{proposition}
A proof is given in the Appendix.

\subsubsection{Selection of the tuning parameter}\label{sec:tuning}
% In practice a tuning parameter $\lambda$ has to be chosen in order to
% tradeoff goodness-of-fit and model complexity. A possibility
% is to use one of the following modified BIC criteria:
% \[\textrm{$BIC_{\mathrm{obs}}$}=-2\ell(\hat{\mu},\hat{K};X_{\mathrm{obs}})+\log(n)\textrm{df},\]
% \[\textrm{$BIC_{compl}$}=-2\mathop{Q}(\hat{\mu},\hat{K}|\hat{\mu},\hat{K})+\log(n)\textrm{df}.\]
% Here $(\hat{\mu},\hat{K})$ denotes the \emph{MissGLasso} estimator
% (\ref{eq:plik}) using the tuning parameter $\lambda$ and
% $\textrm{df}=\sum_{j\leq j'}1_{\{\hat{k}_{j,j'}\neq0\}}$ are the degrees of
% freedom \citep{yuan05model}. \textrm{$BIC_{\mathrm{obs}}$} is based on the observed
% log-likelihood $\ell(\mu,K;X_{\mathrm{obs}})$, which can be computed by
% formula~\ref{eq:observedloglik1}. Criterion \textrm{$BIC_{compl}$} depends
% solely on the $\mathop{Q}$-function and is therefore a by-product of the
% EM algorithm. Ibrahim et al. recently compared BIC criteria based on the
% observed log-likelihood in a missing data context. They find that the
% latter is computationally easier ...
In practice a tuning parameter $\lambda$ has to be chosen in order to
tradeoff goodness-of-fit and model complexity. One possibility
is to use a modified BIC criterion which minimizes
\[\textrm{BIC}=-2\ell(\hat{\mu},\hat{K};\mathbf{x}_{\mathrm{obs}})+\log(n)\textrm{df},\]
over a grid of candidate values for $\lambda$. Here $(\hat{\mu},\hat{K})$
denotes the \emph{MissGLasso} estimator (\ref{eq:plik}) using the tuning
parameter $\lambda$ and $\textrm{df}=\sum_{j\leq j'}1_{\{\hat{K}_{jj'}\neq0\}}$ are the
degrees of freedom \citep{yuan05model}. The defined BIC criterion is based on the observed
log-likelihood $\ell(\mu,K;\mathbf{x}_{\mathrm{obs}})$ which is also suggested by
\cite{ibrahim2008}. 

Another possibility to tune $\lambda$ is to use the popular V-fold
cross-validation method, based on the observed negative log-likelihood as
loss function. We proceed as follows: First divide all the samples into V
disjoint subgroups (folds), and denote the samples in $v$th fold by $N_{v}$ for
$v=1,\ldots,V$. The V-fold cross-validation score is defined as:
\begin{align*}
CV(\lambda)=&\sum_{v=1}^{V}\! \bigg(\!\sum_{i \in
    N_{v}}\!\!\log\!|(\hat{\Sigma}_{-v})_{\mathrm{obs},i}|+(x_{\mathrm{obs},i}\!-(\hat{\mu}_{-v}\!)_{\mathrm{obs},i})^T((\!\hat{\Sigma}_{-v})_{\mathrm{obs},i})^{-1}(x_{\mathrm{obs},i}-(\hat{\mu}_{-v})_{\mathrm{obs},i})\bigg)
\end{align*}
where $\hat{\Sigma}_{-v}=(\hat{K}_{-v})^{-1}$, $\hat{K}_{-v}$ and $\hat{\mu}_{-v}$ denote the
estimates based on the sample $(\cup_{v'=1}^{V}N_{v'})/N_{v}$. Then,
find the best $\hat{\lambda}$ that minimizes $CV(\lambda)$. Finally, fit the
\emph{MissGLasso} to all the data using $\hat{\lambda}$ to get the final
estimator of the inverse covariance matrix.

\section{Extension to sparse regression}\label{sec:sparseregr}
The \emph{MissGLasso} could be applied directly to high-di\-men\-sion\-al regression
with missing values. Suppose a scalar response variable $Y$ is regressed on $p$ predictor
variables $X^{(1)},\ldots,X^{(p)}$. If we assume
joint multivariate normality for $\widetilde{X}=(Y,X^{(1)},\ldots,X^{(p)})$ with mean and
concentration matrix given by
\begin{eqnarray*}
\tilde{\mu}=(\tilde{\mu}_{y},\tilde{\mu}_{x}),\qquad \widetilde{K}=\left(\begin{array}{cc}
\tilde{k}_{yy}&\tilde{k}_{yx}\\
\tilde{k}_{yx}^{T}&\widetilde{K}_{xx}\end{array}\right),
\end{eqnarray*}
we can estimate $(\tilde{\mu},\widetilde{K})$ with the \emph{MissGLasso}. The
regression coefficients $\hat{\beta}$ are then given by
$\hat{\beta}=-\hat{\tilde{k}}_{yy}^{-1}\hat{\tilde{k}}_{yx}$. 
%which is a consequence of formula (\ref{eq:cond2}). 
This approach is short-sighted: a
zero in the concentration matrix, say ${\widetilde{K}}_{jj'}=0$, means that $\widetilde{X}^{(j)}$
and $\widetilde{X}^{(j')}$ are conditionally independent given all other variables
in $\widetilde{X}$, where $Y$ is included in $\widetilde{X}$. But we typically care
about conditional independence of $X^{(j)}$ and $X^{(j')}$ given all
other variables in $X$ (which does not include $Y$). In other words, we think
that sparsity in the concentration matrix $K$ of $X$ (and of course
$\beta$) is desirable. However, sparsity in the matrix $K$ is not enforced by penalizing
$\|\widetilde{K}\|_1$. This can be seen by noting that
$\widehat{K}=\big(\widehat{\widetilde{\Sigma}}_{xx}\big)^{-1}$ is not sparse
for most cases of sparse estimates $\widehat{\widetilde{K}}$. For a similar discussion about this issue, see \cite{scout2009}.

We describe in Section~\ref{sec:two-stage} a two-stage procedure which
results in sparse estimates for the concentration matrix $K$ of $X$ and the
regression parameters $\beta$. In order to motivate the
second stage of this procedure, we first introduce a likelihood-based method
for sparse regression with complete data.   

\subsection{$\ell_{1}$-penalization in the regression model with complete
  data}\label{sec:lassocomplete}
Consider a Gaussian linear model:
\begin{eqnarray*}
 &&Y_i=\beta^{T}X_i+\epsilon_i,\quad i=1,\ldots, n,\\
&&\epsilon_1,\ldots,\epsilon_n \quad\textrm{i.i.d.}\sim\calN(0,\sigma^{2}),
 \end{eqnarray*}
where $X_{i}\in\R^{p}$ are covariates.

In the usual linear regression model, the $\ell_{1}$-norm penalized
estimator, called the Lasso (\cite{tibshirani96regression}), is defined as:
\begin{eqnarray}\label{eq:lasso}
\hat{\beta}_{\lambda}&=&\argmin_{\beta}\frac{1}{2}\|\mathbf{y}-\mathbf{x}\beta\|^{2}+\lambda\|\beta\|_{1},
\end{eqnarray}
with $n\times 1$ vector $\mathbf{y}$, $p\times 1$ regression vector $\beta$ and
$n\times p$ design matrix $\mathbf{x}$. The Lasso estimator in (\ref{eq:lasso}) is not
likelihood-based and does not provide an estimate of the nuisance parameter $\sigma$.
%This estimator is not built upon the Gaussian assumption and does not
%provide an estimate of the nuisance parameter $\sigma$. Nevertheless
%for fixed $\sigma$, the estimator \ref{eq:lasso} can be viewed as minimizing the penalized
%negative log-likelihood $\ell_{\beta}(Y|X)$ w.r.t $\beta$. To get a good
%estimate of $\sigma$ it's natural to take $\sigma$ into the definition and
%optimization of the penalized maximum likelihood estimator: we could
%proceed with the following estimator,
%\begin{eqnarray}\label{eq:lasso2}
%\hat{\beta}_{\lambda},\hat{\sigma}_{\lambda}&=&\argmin_{\beta,\sigma}\left(-\ell_{\beta}(Y|X)+\lambda\|\beta\|_{1}\right)\nonumber\\
%&=&\argmin_{\beta,\sigma}\left(n\log(\sigma)+\frac{1}{2\sigma^{2}}\|Y-X\beta\|^{2}+\lambda\|\beta\|_{1}\right)
%\end{eqnarray}
In \cite{fmrlasso2009}, we suggest to take $\sigma$ into the definition and
optimization of a penalized likelihood estimator: we proceed with the following estimator,
\begin{align}\label{eq:classo}
\hat{\beta}_{\lambda},\hat{\sigma}_{\lambda}=&\argmin_{\beta,\sigma}-\ell(\beta,\sigma;\mathbf{y}|\mathbf{x})+\lambda\frac{\|\beta\|_{1}}{\sigma}\nonumber\\
=&\argmin_{\beta,\sigma}\Big(n\log(\sigma)\!+\frac{1}{2\sigma^{2}}\!\|\mathbf{y}-\!\mathbf{x}\beta\|^{2}\!+\lambda\frac{\|\beta\|_{1}}{\sigma}\Big).
\end{align}
%The re-formulation of the estimator in terms of the inner products of the data
%$Y$ and $X$, will be useful in the next section.
Intuitively the estimator (\ref{eq:classo}) penalizes the $\ell_{1}$-norm of
the regression coefficients and small variances $\sigma$
simultaneously. Furthermore this estimator is equivariant under scaling
(see \cite{fmrlasso2009}). Most
importantly if we reparametrize $\rho=1/\sigma$ and $\phi=\beta/\sigma$ we get the
following convex optimization problem:
\begin{align}\label{eq:classo2}
\hat{\phi}_{\lambda},\hat{\rho}_{\lambda}=&\argmin_{\phi,\rho}\Big(\!-n\log(\rho)+\!\frac{1}{2}\!\|\rho\mathbf{y}-\!\mathbf{x}\phi\|^{2}+\!\lambda\|\phi\|_{1}\Big).
\end{align}
This optimization problem can be solved efficiently in a coordinate-wise
fashion. The following algorithm is very easy to implement, it simply
updates, in each iteration, $\rho$ followed by the coordinates $\phi_j$, $j=1,\ldots,p$, of $\phi$.

\textbf{Coordinate-wise algorithm for solving (\ref{eq:classo2})}
\begin{itemize}
\item [1.] Start with initial guesses for $\phi^{(0)},\rho^{(0)}$.
\item [2.] Update the current estimates $\phi^{(m)},\rho^{(m)}$
  coordinate-wise by:
  \begin{eqnarray*}
    \rho^{(m+1)}&=&\frac{\mathbf{y}^{T}\mathbf{x}\phi^{(m)}+\sqrt{(\mathbf{y}^{T}\mathbf{x}\phi^{(m)})^{2}+4\mathbf{y}^{T}\mathbf{y}n}}{2\mathbf{y}^{T}\mathbf{y}}\\
    \phi_{j}^{(m+1)}&=&\left\{\begin{array}{ll}
        0&\textrm{if $|S_{j}|\leq\lambda$}\\
        (\lambda-S_{j})/\mathbf{x}_{j}^{T}\mathbf{x}_{j}&\textrm{if
          $S_{j}>\lambda$}\\
        -(\lambda+S_{j})/\mathbf{x}_{j}^{T}\mathbf{x}_{j}&\textrm{if
          $S_{j}<-\lambda$}\end{array}\right.
  \end{eqnarray*}
  where $S_{j}$ is defined as $$S_{j}=-\rho^{(m+1)}\mathbf{x}_j^{T}\mathbf{y}+\sum\limits_{s<j}\phi_{s}^{(m+1)}\mathbf{x}_{j}^{T}\mathbf{x}_{s}
  +\sum\limits_{s>j}\phi_{s}^{(m)}\mathbf{x}_{j}^{T}\mathbf{x}_{s}$$
  and $j=1,\ldots,p$.
\item [3.] Iterate step 2 until convergence.
\end{itemize}

With $\mathbf{x}_{j}$ we denote the $j$th column vector of the $n\times p$
matrix $\mathbf{x}$. This algorithm can be implemented very efficiently as
it is the case for the coordinate descent algorithm solving the usual Lasso problem. For
example \emph{naive updates}, \emph{covariance updates} and the
\emph{active-set} strategy described in
\cite{friedman07fastlasso} and \cite{friedmanetal08} are applicable here as well.

Numerical convergence of the above algorithm is ensured as follows.

\begin{proposition}\label{prop:convergence2}
 Every limit point $(\bar{\rho},\bar{\phi})$ of the sequence $\{(\rho^{(m)},\phi^{(m)});m =
 0,1,2,\ldots\}$, generated by the above algorithm, is a stationary point
 of the criterion function in (\ref{eq:classo2}).
\end{proposition}
A proof is given in the Appendix.

Note that the algorithm only involves inner products
of $\mathbf{x}$ and $\mathbf{y}$. We will make use of this algorithm in the
next section when treating regression with missing values.

\subsection{Two-stage likelihood approach for sparse regression with
  missing data}\label{sec:two-stage}
We now develop a two-stage $\ell_{1}$-penalized likelihood approach for
sparse regression with potential missing values in the design matrix
$\mathbf{x}$. Consider the Gaussian linear model:
\begin{eqnarray}\label{eq:normalregr}
 &&X_i\sim\calN(\mu,\Sigma),\quad X_i=(X^{(1)}_{i},\ldots,X^{(p)}_{i})\in\R^{p}\nonumber\\
 &&Y_i|X_i=\beta^{T}X_i+\epsilon_i,\quad\epsilon_i\; \textrm{i.i.d.}
 \sim\calN(0,\sigma^{2})\\
&&X_i,\:\epsilon_i\quad \textrm{independent of each other and
  among}\quad i=1,\ldots,n.\nonumber 
\end{eqnarray}
If we assume model (\ref{eq:normalregr}) it is obvious that $(Y_i,X_i)$ follows again
a multivariate normal distribution. The corresponding mean and covariance matrix are
given in the following lemma:
\begin{lemma}\label{lemma:mvn}
 Assuming model (\ref{eq:normalregr}), $(Y_i,X_i)$ is normally distributed
 $\calN(\tilde{\mu},\widetilde{\Sigma})$ with $\tilde{\mu}=(\beta^{T}\mu,\mu)$ and
 \begin{eqnarray}
\widetilde{\Sigma}=\left(\begin{array}{cc}\sigma^{2}\!+\!\beta^{T}\!\Sigma\beta&\beta^{T}\!\Sigma\\
 \Sigma\beta&\Sigma\end{array}\right),\;\widetilde{K}=\left(\begin{array}{cc}\frac{1}{\sigma^{2}}&-\frac{\beta^{T}}{\sigma^{2}}\\
 -\frac{\beta}{\sigma^{2}}&K\!+\!\frac{\beta\beta^{T}}{\sigma^{2}}\end{array}\right)
 \end{eqnarray}
 % \begin{eqnarray}
 % \tilde{\mu}=(\beta^{T}\mu,\mu),\quad \tilde{\Sigma}=\left(\begin{array}{cc}\sigma^{2}\!+\!\beta^{T}\!\Sigma\beta&\beta^{T}\!\Sigma\\
 % \Sigma\beta&\Sigma\end{array}\right),\quad\tilde{K}=\left(\begin{array}{cc}\frac{1}{\sigma^{2}}&-\frac{\beta^{T}}{\sigma^{2}}\\
 % -\frac{\beta}{\sigma^{2}}&K\!+\!\frac{\beta\beta^{T}}{\sigma^{2}}\end{array}\right)
 % \end{eqnarray}
\end{lemma} 
A proof is given in the Appendix.

In a first stage of the procedure we estimate the inverse
covariance $K=\Sigma^{-1}$ of $X$ using the \emph{MissGLasso}:
\\ \\
\textbf{1st stage:}
\begin{eqnarray}\label{eq:1ststage}
\hat{\mu}_{\lambda_1},\hat{K}_{\lambda_{1}}=\argmin_{(\mu,K):K\succ0}-\ell(\mu,K;\mathbf{x}_{\mathrm{obs}})\!+\!\lambda_{1}\|K\|_{1}.
\end{eqnarray}
Let now $\ell(\beta,\sigma,\mu,K;\mathbf{y,x_{\mathrm{obs}}})$ be the observed
log-like\-li\-hood of the data $(\mathbf{y,x})$. In the second stage of the
procedure we hold $\mu$ and $K$ fixed at the values
$\hat{\mu}_{\lambda_{1}}$ and $\hat{K}_{\lambda_{1}}$ from the first stage
and estimate $\beta$ and $\sigma$ by: 
\\ \\
\textbf{2nd stage:}
\begin{align}\label{eq:2ndstage}
&\hat{\beta}_{\lambda_2},\hat{\sigma}_{\lambda_{2}}=\argmin_{\beta,\sigma}-\ell(\beta,\sigma,\hat{\mu}_{\lambda_{1}},\hat{K}_{\lambda_{1}};\mathbf{y,x_{\mathrm{obs}}})\!+\!\lambda_{2} \frac{\|\beta\|_{1}}{\sigma}.
\end{align} 
Note that we use two different tuning parameters for the first and the second stage, denoted by $\lambda_{1}$ and $\lambda_{2}$. In practice, instead of tuning over a two-dimensional grid
$(\lambda_1,\lambda_2)$, we consider the 1st and 2nd stage
independently. We tune first $\lambda_{1}$ using BIC or
cross-validation as explained in Section~\ref{sec:tuning} and then we use the
resulting estimator in the 2nd stage and tune $\lambda_2$.

A detailed description of the EM algorithm for solving the 1st stage
problem was given in Section~\ref{subsec:emalg}. We now present an
EM algorithm for solving the 2nd stage. In the E-Step of our algorithm,
we calculate the conditional expectation of the complete-data
log-likelihood given by
\begin{eqnarray}\label{eq:complloglik2}
\ell(\beta,\sigma,\hat{\mu}_{\lambda_{1}},\hat{K}_{\lambda_{1}};\mathbf{y,x})&&=\ell(\beta,\sigma;\mathbf{y}|\mathbf{x})+\ell(\hat{\mu}_{\lambda_{1}},\hat{K}_{\lambda_{1}};\mathbf{x})\nonumber\\
&&=
\ell(\beta,\sigma;\mathbf{y}|\mathbf{x})+\mathrm{const}\nonumber\\
&&=-n\log(\sigma)-\frac{1}{2\sigma^{2}}\|\mathbf{y}-\mathbf{x}\beta\|^{2}+\mathrm{const}\\
&&=-n\log(\sigma)-\left(\frac{\mathbf{y}^{T}\mathbf{y}}{2\sigma^{2}}-\frac{\mathbf{y}^{T}\mathbf{x}\beta}{\sigma^{2}}+\frac{\beta^{T}\mathbf{x}^{T}\mathbf{x}\beta}{2\sigma^{2}}\right)+\mathrm{const}.\nonumber
\end{eqnarray}
% \begin{eqnarray}\label{eq:complloglik2}
% \ell(\beta,\sigma,\hat{\mu}_{\lambda_{1}},\hat{K}_{\lambda_{1}};\mathbf{y,x})&=&\ell(\beta,\sigma;\mathbf{y}|\mathbf{x})+\ell(\hat{\mu}_{\lambda_{1}},\hat{K}_{\lambda_{1}};\mathbf{x})\nonumber\\
% &=&
% \ell(\beta,\sigma;\mathbf{y}|\mathbf{x})+\mathrm{const}\nonumber\\
% &=&-n\log(\sigma)-\frac{1}{2\sigma^{2}}\|\mathbf{y}-\mathbf{x}\beta\|^{2}+\mathrm{const}\\
% &=&-n\log(\sigma)-\left(\frac{\mathbf{y}^{T}\mathbf{y}}{2\sigma^{2}}-\frac{\mathbf{y}^{T}\mathbf{x}\beta}{\sigma^{2}}+\frac{\beta^{T}\mathbf{x}^{T}\mathbf{x}\beta}{2\sigma^{2}}\right)+\mathrm{const}.\nonumber
% \end{eqnarray}
We see from Equation (\ref{eq:complloglik2}) that the part of the complete
log-likelihood which depends only on the regression parameters $\beta$ and
$\sigma$ is linear in the inner products $\mathbf{y}^{T}\mathbf{y}$,
$\mathbf{y}^{T}\mathbf{x}$ and $\mathbf{x}^{T}\mathbf{x}$. Therefore we can
write the E-Step as:
\\ \\
\textbf{E-Step:}
 \begin{eqnarray}
 \mathbf{T}^{(m+1)}_{1}&=&\E[\mathbf{y}^{T}\mathbf{x}|\mathbf{y},\mathbf{x}_{\mathrm{obs}},\beta^{(m)},\sigma^{(m)},\hat{\mu}_{\lambda_{1}},\hat{K}_{\lambda_{1}}]\nonumber\\
 \mathbf{T}^{(m+1)}_{2}&=&\E[\mathbf{x}^{T}\mathbf{x}|\mathbf{y},\mathbf{x}_{\mathrm{obs}},\beta^{(m)},\sigma^{(m)},\hat{\mu}_{\lambda_{1}},\hat{K}_{\lambda_{1}}]\nonumber.
 \end{eqnarray}
These conditional expectations can be computed as in
Section~\ref{subsec:emalg} using Lemma~\ref{lemma:mvn}. In particular,
these computations involve inversion of the matrices
$\widetilde{K}^{(m)}_{\mathrm{mis},\mathrm{mis}}$. Because of the special structure of
$\widetilde{K}^{(m)}_{\mathrm{mis},\mathrm{mis}}$, see Lemma~\ref{lemma:mvn}, explicit inversion is
possible by exploiting the formula
$(A+bb^{T})^{-1}=A^{-1}-A^{-1}bb^{T}A^{-1}/(1+b^{T}A^{-1}b)$, where
$A^{-1}$ has been previously computed in the first stage.
%$(A+bb^{T})^{-1}=A^{-1}-A^{-1}bb^{T}A^{-1}/(1+b^{T}A^{-1}b)$.

Finally, in the M-Step, we update the regression coefficients by:
\\ \\
\textbf{M-Step:}
\begin{align}\label{eq:mstep2}
 \beta^{(m+1)},\sigma^{(m+1)}=&\argmin_{\beta,\sigma}\bigg(\!
 n\log(\sigma)\!+\!\frac{\mathbf{y}^{T}\mathbf{y}}{2\sigma^{2}}\!-\!\frac{\mathbf{T}^{(m+1)}_{1}\beta}{\sigma^{2}}+\frac{\beta^{T}\mathbf{T}^{(m+1)}_{2}\beta}{2\sigma^{2}}+\lambda
 \frac{\|\beta\|_{1}}{\sigma}\bigg).
\end{align}
If we reparametrize $\rho=1/\sigma$ and $\phi=\beta/\sigma$ in
(\ref{eq:mstep2}), we see that the M-Step has essentially the same form as
(\ref{eq:classo2}). Therefore, we can use the algorithm described in
Section~\ref{sec:lassocomplete} but exchanging the
inner products $\mathbf{y}^{T}\mathbf{x}$ and
$\mathbf{x}^{T}\mathbf{x}$ for $\mathbf{T}^{(m+1)}_1$ and $\mathbf{T}^{(m+1)}_2$.

\section{Simulations}
\subsection{Simulations for sparse inverse covariance estimation}
\subsubsection{Simulation 1}\label{sec:simulation1}
We consider model 1, model 2, model 3 and model 4 of \cite{rothman2008spi} with p = 10, 50,
100, 200, 300: $X_{1},\ldots,X_{n}$ i.i.d. $\sim\calN(0,\Sigma)$ with 
\begin{description}
\item[Model 1:] $n=100$. AR(1), $\Sigma_{jj'}=0.7^{|j'-j|}$.\vspace{0.15cm}

\item[Model 2:] $n=150$. AR(4), $K_{jj'}\!=\!\mathrm{I}_{(|j'-j|=0)}+0.4\mathrm{I}_{(|j'-j|=1)}+0.2\mathrm{I}_{(|j'-j|=2)}+0.2\mathrm{I}_{(|j'-j|=3)}+0.1\mathrm{I}_{(|j'-j|=4)}$.\vspace{0.15cm}

\item[Model 3:]$n=200$. $K=B+\delta\mathrm{I}$, where each off-diagonal entry in $B$ is
generated independently and equals $0.5$ with probability $\alpha=0.1$ or 0 with
probability $1-\alpha=0.9$, all diagonal  entries of $B$ are zero, and $\delta$ is
chosen such that the condition number of $K$ is~$p$.\vspace{0.15cm}

\item[Model 4:]$n=250$. Same as model 3 except $\alpha=0.5$.
\end{description}

% \begin{tabular}{ll}
%  Model 1:& $n=100$. AR(1), $\Sigma_{jj'}=0.7^{|j'-j|}$.\vspace{0.2cm}\\
%  Model 2:& $n=150$. AR(4), $K_{jj'}=\mathrm{I}_{(|j'-j|=0)}+0.4\mathrm{I}_{(|j'-j|=1)}+$\\ & \hspace*{3.7cm}$0.2\mathrm{I}_{(|j'-j|=2)}+0.2\mathrm{I}_{(|j'-j|=3)}+0.1\mathrm{I}_{(|j'-j|=4)}$.
% \vspace{0.2cm}\\
% Model 3: &$n=200$. $K=B+\delta\mathrm{I}$, where each off-diagonal entry in $B$ is
% generated independently\\ & and equals $0.5$ with probability $\alpha=0.1$ or 0 with
% probability $1-\alpha=0.9$, all diagonal \\ & entries of $B$ are zero, and $\delta$ is
% chosen such that the condition number of $K$ is $p$.\vspace{0.2cm}\\
% Model 4:& $n=250$. Same as model 3 except $\alpha=0.5$.\vspace{0.2cm}
% \end{tabular}
Note that in all models $\Sigma^{-1}$ is sparse. In models 1 and 2 the number of non-zeros
in $\Sigma^{-1}$ is linear in $p$, whereas in models 3 and 4 it is proportional to $p^2$.

For all $20$ settings (4 models with $p=10$, $50$, $100$, $200$, $300$) we make 50 simulation runs. In each run we proceed as
follows: 
\begin{itemize}
\item
We generate $n$ training observations
and a separate set of $n$ validation observations.
\item
In the training set we delete completely at random $10\%, 20\%$
and $30\%$ of the data. Per setting, we therefore get three training sets with
different degree of missing data.
\item
The \emph{MissGLasso} estimator is fitted on each of the three mutilated
training sets, with the tuning
parameter $\lambda$ selected by minimizing twice the negative
log-likelihood (log-loss) on the validation data. This results in three different estimators of the
concentration matrix $K$. 
\end{itemize}
We evaluate the concentration matrix estimation performance using the
Kullback-Leibler loss:
\begin{equation*}\label{eq:klloss}
\Delta_{KL}(\hat{K},K)=\mathrm{tr}(\Sigma\hat{K})-\log|\Sigma\hat{K}|-p.
\end{equation*}
We compare the \emph{MissGLasso} with the following
estimators: 
\begin{itemize}
\item \emph{MeanImp}: Impute the missing values by their corresponding
  column means. Then apply the GLasso from~(\ref{eq:glasso}) on the imputed data. 
\item \emph{MissRidge}: Estimate $\hat{K}=\hat{\Sigma}^{-1}$ by
minimizing $-\ell(\mu,K;\mathbf{x}_{\mathrm{obs}})+\lambda\|K\|_2^2.$ For
optimization we use an EM algorithm with an $\ell_2$-penalized (inverse) covariance update in the
M-Step. In the case of complete data, covariance estimation with an
$\ell_2$-penalty is derived in \cite{scout2009}. 
\item \emph{MLE}: Compute the (unpenalized) maximum likelihood
estimator using the EM algorithm implemented in the R-package
\emph{norm} (only for $p=10$).
\end{itemize}

Results for all covariance models with different degrees of
missingness are summarized in Tables~\ref{tab:KLloss1} and \ref{tab:KLloss2} which report the average
Kullback-Leibler loss and the standard error. For all settings of
models 1 and 3 the \emph{MissGLasso} outperforms \emph{MeanImp} and
\emph{MissRidge} significantly. In model 2 \emph{MissGLasso} works
competitive but sometimes \emph{MeanImp} or \emph{MissRidge} is slightly better. In model 4, the most dense scenario,
\emph{MissRidge} exhibits the lowest average Kullback-Leibler
loss. Interestingly, in models 1 and 2 with large values of $p$,
\emph{MissRidge} works rather poorly in comparison to \emph{MeanImp}. The
reason is that in very sparse settings the gain of
$\ell_1$- over $\ell_2$-regularization dominates the gain of EM-type
estimation over ``naive'' column-wise mean imputation. For the lowest dimensional
case ($p=10$) we further notice that the \emph{MLE} estimator performs very badly
with high degrees of missingness whereas the \emph{MissGLasso}
and the \emph{MissRidge} remain stable.

\begin{table}[!p]
\centering
\tabcolsep=1.5pt
\begin{tabular}{|cc||c|c|c|c|} \hline
\multicolumn{2}{|c||}{\textbf{Model 1}}&MLE&MeanImp&MissRidge&MissGLasso\\\hline
p=10&10\%&0.82(0.03) & 0.66(0.02) & 0.53(0.02) & \textbf{0.41(0.02)} \\ \hline
    &20\%&1.34(0.07) & 1.04(0.03) & 0.66(0.02) & \textbf{0.50(0.02)} \\ \hline
    &30\%&3.32(0.39) & 1.60(0.05) & 0.79(0.02) & \textbf{0.61(0.02)} \\\hline
p=50&10\%&NA&6.49(0.06) & 9.39(0.06) & \textbf{4.81(0.04)} \\ \hline
&20\%&NA &9.17(0.10) & 10.84(0.08) & \textbf{5.63(0.06)} \\ \hline
&30\%&NA&12.38(0.10) & 12.44(0.09) & \textbf{6.62(0.07)} \\ \hline
p=100&10\%&NA& 16.49(0.10) & 29.79(0.12) & \textbf{13.07(0.08)} \\ \hline
 &20\%&NA& 21.77(0.12) & 33.25(0.13) & \textbf{14.99(0.10)} \\ \hline
&30\%&NA &28.65(0.20) & 37.35(0.14) & \textbf{17.72(0.12)} \\ \hline
p=200&10\%&NA& 40.36(0.14) & 85.83(0.15) & \textbf{33.79(0.14)} \\ \hline
&20\%&NA &50.61(0.18) & 92.52(0.15) & \textbf{38.13(0.14)} \\ \hline
&30\%&NA &64.35(0.27) & 100.03(0.14) & \textbf{44.66(0.18)} \\ \hline
p=300&10\%&NA& 67.20(0.14) & 151.85(0.15) & \textbf{57.95(0.14)} \\ \hline
&20\%&NA&82.39(0.26) & 160.85(0.16) & \textbf{65.13(0.17)} \\ \hline
&30\%&NA&103.03(0.26) & 170.22(0.14) & \textbf{75.46(0.21)} \\ \hline
\end{tabular}\vspace{0.1cm}

\begin{tabular}{|cc||c|c|c|c|}\hline
\multicolumn{2}{|c||}{\textbf{Model 2}}&MLE&MeanImp&MissRidge&MissGLasso\\\hline
p=10&10\%& 0.53(0.02) & 0.50(0.01) & \textbf{0.42(0.01)} & 0.44(0.01)\\\hline
&20\%& 0.72(0.03) & 0.75(0.02) & \textbf{0.48(0.01)} & 0.51(0.01) \\ \hline
&30\%& 1.29(0.07) & 1.25(0.03) & \textbf{0.64(0.02)} & 0.65(0.02) \\ \hline
p=50&10\%& NA & \textbf{4.31(0.03)} & 6.27(0.02) & 4.33(0.02) \\ \hline
&20\%& NA & 5.32(0.04) & 6.86(0.02) & \textbf{4.84(0.03)} \\ \hline
&30\%& NA & 7.43(0.05) & 7.49(0.03) & \textbf{5.52(0.04)} \\ \hline
p=100&10\%& NA & \textbf{9.66(0.04)} & 17.12 (0.03) & 9.93 (0.04) \\ \hline
&20\%& NA & 11.56(0.06) & 18.05(0.03) & \textbf{11.08(0.04)} \\ \hline
&30\%& NA & 15.33(0.06) & 18.87(0.03) & \textbf{12.28(0.04)} \\\hline
p=200&10\%& NA & \textbf{21.36(0.08)} & 43.46(0.04) & 22.28(0.07) \\ \hline
&20\%& NA & \textbf{24.61(0.10)} & 44.33(0.04) & 24.72(0.07) \\ \hline
&30\%& NA & 31.34(0.06) & 45.15(0.04) & \textbf{27.26(0.06)} \\ 
\hline
p=300&10\%& NA & \textbf{33.48(0.06)} & 71.98(0.05) & 35.44(0.06) \\ \hline
&20\%& NA & \textbf{38.42(0.09)} & 72.38(0.05) & 38.88(0.08) \\ \hline
&30\%& NA & 47.37(0.02) & 72.72(0.05) & \textbf{43.14(0.07)} \\ \hline
\end{tabular}
\caption{Model 1 and Model 2 (strong sparsity): Average (SE) Kullback-Leibler loss of \emph{MLE},
  \emph{MeanImp}, \emph{MissRidge} and \emph{MissGLasso} with different degrees of
  missingness. Method with lowest average Kullback-Leibler loss in bold face}
\label{tab:KLloss1}
\end{table}

\begin{table}[!p]
\centering
\tabcolsep=1.5pt
\begin{tabular}{|cc||c|c|c|c|}\hline
\multicolumn{2}{|c||}{\textbf{Model 3}}&MLE&MeanImp&MissRidge&MissGLasso\\\hline
p=10&10\%& 0.38(0.01) & 0.31(0.01) & 0.30(0.01) & \textbf{0.22(0.01)} \\ \hline
&20\%& 0.51(0.02) & 0.53(0.01) & 0.36(0.01) & \textbf{0.26(0.01)} \\ \hline
&30\%& 0.78(0.03) & 0.98(0.02) & 0.45(0.01) & \textbf{0.33(0.01)} \\ \hline
p=50&10\%& NA & 3.56(0.03) & 4.71(0.02) & \textbf{3.04(0.02)} \\ \hline
&20\%& NA & 5.05(0.04) & 5.30(0.03) & \textbf{3.63(0.03)} \\ \hline
&30\%& NA & 7.36(0.07) & 5.98(0.03) & \textbf{4.41(0.04)} \\ \hline
p=100&10\%& NA & 10.45(0.05) & 13.86(0.04) & \textbf{9.53(0.05)} \\ \hline
&20\%& NA & 13.41(0.07) & 15.06(0.04) & \textbf{11.05(0.06)} \\ \hline
&30\%& NA & 18.15(0.10) & 16.42(0.05) &\textbf{ 13.01(0.06)} \\ \hline
p=200&10\%& NA & 31.92(0.08) & 38.97(0.05) & \textbf{30.74(0.07)} \\ \hline
&20\%& NA & 37.49(0.11) & 41.13(0.06) & \textbf{34.23(0.09)} \\ \hline
&30\%& NA & 46.18(0.16) & 43.67(0.06) & \textbf{38.15(0.08)} \\ \hline
p=300&10\%& NA & 60.69(0.10) & 71.39(0.07) & \textbf{59.13(0.10)} \\ \hline
&20\%& NA & 69.60(0.16) & 74.92(0.08) & \textbf{64.98(0.12)} \\ \hline
&30\%& NA & 83.12(0.19) & 79.39(0.08) & \textbf{71.58(0.11)} \\ \hline
\end{tabular}\vspace{0.1cm}

\begin{tabular}{|cc||c|c|c|c|}\hline
\multicolumn{2}{|c||}{\textbf{Model 4}}&MLE&MeanImp&MissRidge&MissGLasso\\\hline
p=10&10\%& 0.30(0.01) & 0.29(0.01) & 0.24(0.01) & \textbf{0.23(0.01)} \\ \hline
&20\%& 0.40(0.01) & 0.54(0.02) & 0.30(0.01) & \textbf{0.29(0.01)} \\ \hline
&30\%& 0.56(0.02) & 0.94(0.02) & \textbf{0.36(0.01)} & 0.37(0.01) \\\hline
p=50&10\%& NA & 5.23(0.03) & \textbf{4.27(0.02)} & 5.04(0.03) \\ \hline
&20\%& NA & 6.66(0.04) &\textbf{ 4.88(0.03)} & 5.77(0.03) \\ \hline
&30\%& NA & 8.95(0.07) &\textbf{ 5.50(0.03)} & 6.55(0.04) \\ \hline
p=100&10\%& NA & 14.23(0.04) & \textbf{12.69(0.03)} & 14.02(0.04) \\ \hline
&20\%& NA & 16.79(0.06) & \textbf{13.93(0.03)} & 15.37(0.04) \\\hline
&30\%& NA & 21.27(0.10) & \textbf{15.25(0.05)} & 16.83(0.05) \\ \hline
p=200&10\%& NA & 39.43(0.09) & \textbf{37.00(0.07)} & 39.11(0.08) \\ \hline
&20\%& NA & 44.62(0.12) & \textbf{39.51(0.07)} & 42.19(0.08) \\ \hline
&30\%& NA & 53.48(0.19) &\textbf{42.41(0.07)} & 45.64(0.08) \\ \hline
p=300&10\%& NA & 65.44(0.09) & \textbf{65.24(0.07)} & 65.43(0.08) \\\hline
&20\%& NA & 72.43(0.12) & \textbf{68.97(0.06)} & 69.62(0.08) \\ \hline
&30\%& NA & 85.19(0.17) & \textbf{73.59(0.07)} & 74.19(0.09) \\\hline
% p=300&10\%&NA&65.71(0.10)&\textbf{65.40(0.07)}&65.59(0.08)\\\hline
% &20\%&NA&72.76(0.13)&\textbf{69.18(0.07)}&69.83(0.09)\\\hline
% &30\%&NA&85.57(0.20)&\textbf{73.58(0.08)}&74.14(0.10)\\\hline
\end{tabular}
\caption{Model 3 and Model 4 (weak sparsity): Average (SE) Kullback-Leibler loss of \emph{MLE},
  \emph{MeanImp}, \emph{MissRidge} and \emph{MissGLasso} with different degrees of
  missingness. Method with lowest average Kullback-Leibler loss in bold face}
\label{tab:KLloss2}
\end{table}

To assess the performance of \emph{MissGLasso} on recovering the sparsity
structure in $K$, we also report the true positive rate (TPR) and the true
negative rate (TNR) defined as 
\begin{eqnarray*}
\textrm{TPR}&=&\frac{\# \textrm{true non-zeros estimated as non-zeros}}{\# \textrm{true non-zeros}},\\
\textrm{TNR}&=&\frac{\# \textrm{true zeros estimated as zeros}}{\# \textrm{true zeros}}.
\end{eqnarray*}
These numbers are reported in Tables~\ref{tab:TPR1} and \ref{tab:TPR2}. For
visualization, we also plot in Figure~\ref{fig:heatmap} heat-maps of the percentage of times each element
was estimated as zero among the 50 simulation runs. We note that our choice
of CV-optimal $\lambda$ has a tendency to yield too many false positives
and thus too low values for TNR: in the case without missing values, this
finding is theoretically supported in \cite{meinshausen04consistent}.

\begin{table}[!p]
\centering
\tabcolsep=8pt
\begin{tabular}{|cc||c|c|}\hline
\multicolumn{2}{|c||}{\textbf{Model 1}}&TPR [\%]& TNR [\%]\\\hline
p=10&10\%& 100 (0.00) & 39.06 (1.45) \\\hline
&20\%& 100 (0.00) & 42.06 (1.32) \\ \hline
&30\%& 100 (0.00) & 43.94 (1.33) \\ \hline
p=50&10\%& 100 (0.00) & 67.78 (0.34) \\ \hline
&20\%& 100 (0.00) & 67.64 (0.39) \\ \hline
&30\%& 100 (0.00) & 69.78 (0.24) \\ \hline
p=100&10\%& 100 (0.00) & 77.05 (0.23) \\ \hline
&20\%& 100 (0.00) & 77.01 (0.24) \\ \hline
&30\%& 99.99 (0.01) & 78.75 (0.09) \\ \hline
p=200&10\%& 100 (0.00) & 83.89 (0.17) \\ \hline
&20\%& 100 (0.00) & 85.10 (0.04) \\ \hline
&30\%& 99.98 (0.01) & 85.24 (0.15) \\ \hline
p=300&10\%& 100 (0.00) & 87.36 (0.13) \\\hline
&20\%& 100 (0.00) & 88.41 (0.03) \\\hline
&30\%& 100 (0.00) & 88.44 (0.07) \\ \hline
\end{tabular}
\begin{tabular}{|cc||c|c|}\hline
\multicolumn{2}{|c||}{\textbf{Model 2}}&TPR [\%]& TNR [\%]\\\hline
p=10&10\%& 93.14 (1.06) & 21.07 (2.36) \\ \hline
&20\%& 88.46 (1.46) & 25.60 (2.59) \\ \hline
&30\%& 80.51 (1.58) & 36.13 (2.66) \\ \hline
p=50&10\%& 57.75 (0.35) & 74.13 (0.31) \\\hline
&20\%& 53.20 (0.59) & 76.50 (0.60) \\ \hline
&30\%& 49.47 (0.59) & 79.39 (0.55) \\\hline
p=100&10\%& 48.81 (0.29) & 85.01 (0.21) \\ \hline
&20\%& 46.72 (0.41) & 85.35 (0.43) \\ \hline
&30\%& 43.60 (0.25) & 86.94 (0.09) \\ \hline
p=200&10\%& 44.28 (0.13) & 90.40 (0.05) \\ \hline
&20\%& 41.40 (0.35) & 91.26 (0.30) \\ \hline
&30\%& 37.53 (0.15) & 92.41 (0.04) \\ \hline
p=300&10\%& 41.74 (0.25) & 93.21 (0.20) \\ \hline
&20\%& 39.19 (0.12) & 93.47 (0.03) \\ \hline
&30\%& 32.56 (0.16) & 96.04 (0.07) \\ \hline
\end{tabular}
\caption{Model 1 and Model 2 (strong sparsity): Average (SE) of True Positive Rate (TPR) and True
  Negative Rate (TNR) of the \emph{MissGLasso} estimator for inferring the
  zeros in $K=\Sigma^{-1}$. All numbers are percentages}
\label{tab:TPR1}
\end{table}

\begin{table}[!p]
\centering
\tabcolsep=8pt
\begin{tabular}{|cc||c|c|}\hline
\multicolumn{2}{|c||}{\textbf{Model 3}}&TPR [\%]& TNR [\%]\\\hline
p=10&10\%& 100 (0.00) & 43.15 (1.63) \\\hline
&20\%& 100 (0.00) & 44.05 (1.69) \\ \hline
&30\%& 100 (0.00) & 43.50 (1.16) \\ \hline
p=50&10\%& 99.75 (0.06) & 63.55 (0.40) \\ \hline
&20\%& 98.92 (0.14) & 64.86 (0.32) \\ \hline
&30\%& 97.22 (0.20) & 67.12 (0.27) \\ \hline
p=100&10\%& 94.52 (0.14) & 70.92 (0.08) \\ \hline
&20\%& 89.78 (0.20) & 74.47 (0.09) \\ \hline
&30\%& 82.56 (0.25) & 77.93 (0.08) \\ \hline
p=200&10\%& 73.60 (0.15) & 78.06 (0.05) \\ \hline
&20\%& 64.66 (0.17) & 81.20 (0.05) \\ \hline
&30\%& 54.49 (0.17) & 84.17 (0.05) \\ 
\hline
p=300&10\%& 61.19 (0.10) & 82.35 (0.03) \\ \hline
&20\%& 52.47 (0.10) & 84.91 (0.03) \\ \hline
&30\%& 43.19 (0.12) & 87.31 (0.03) \\ \hline
\end{tabular}
\begin{tabular}{|cc||c|c|}\hline
\multicolumn{2}{|c||}{\textbf{Model 4}}&TPR [\%]& TNR [\%]\\\hline
p=10&10\%& 100 (0.00) & 26.50 (1.60) \\ \hline
&20\%& 100 (0.00) & 24.42 (1.45) \\ \hline
&30\%& 99.38 (0.25) & 26.58 (1.68) \\ \hline
p=50&10\%& 80.29 (0.29) & 34.35 (0.36) \\ \hline
&20\%& 72.78 (0.42) & 39.88 (0.43) \\ \hline
&30\%& 64.12 (0.51) & 46.31 (0.51) \\ \hline
p=100&10\%& 54.33 (0.40) & 53.67 (0.39) \\ \hline
&20\%& 47.54 (0.36) & 58.91 (0.37) \\ \hline
&30\%& 40.13 (0.29) & 65.02 (0.31) \\ \hline
p=200&10\%& 36.65 (0.22) & 67.62 (0.22) \\ \hline
&20\%& 31.39 (0.23) & 72.01 (0.23) \\ \hline
&30\%& 26.81 (0.25) & 76.13 (0.25) \\ \hline
p=300&10\%& 26.73 (0.35) & 75.64 (0.34) \\ \hline
&20\%& 23.35 (0.32) & 78.53 (0.32) \\ \hline
&30\%& 20.62 (0.13) & 80.99 (0.14) \\ \hline
% p=300&10\%&27.23 (0.37)&75.14 (0.36)\\\hline
% &20\%&22.82 (0.33)&79.15 (0.32)\\\hline
% &30\%&20.51 (0.16)&81.13 (0.16)\\\hline
\end{tabular}
\caption{Model 3 and Model 4 (weak sparsity): Average (SE) of True Positive Rate (TPR) and True
  Negative Rate (TNR) of the \emph{MissGLasso} estimator for inferring the
  zeros in $K=\Sigma^{-1}$. All numbers are percentages}
\label{tab:TPR2}
\end{table}

\begin{figure}[!p]%--- Bild 'H'ier, 'B'ottom oder 'T'op  ``!'' ich WILL!!
%\begin{centering}
 \includegraphics[scale=0.78]{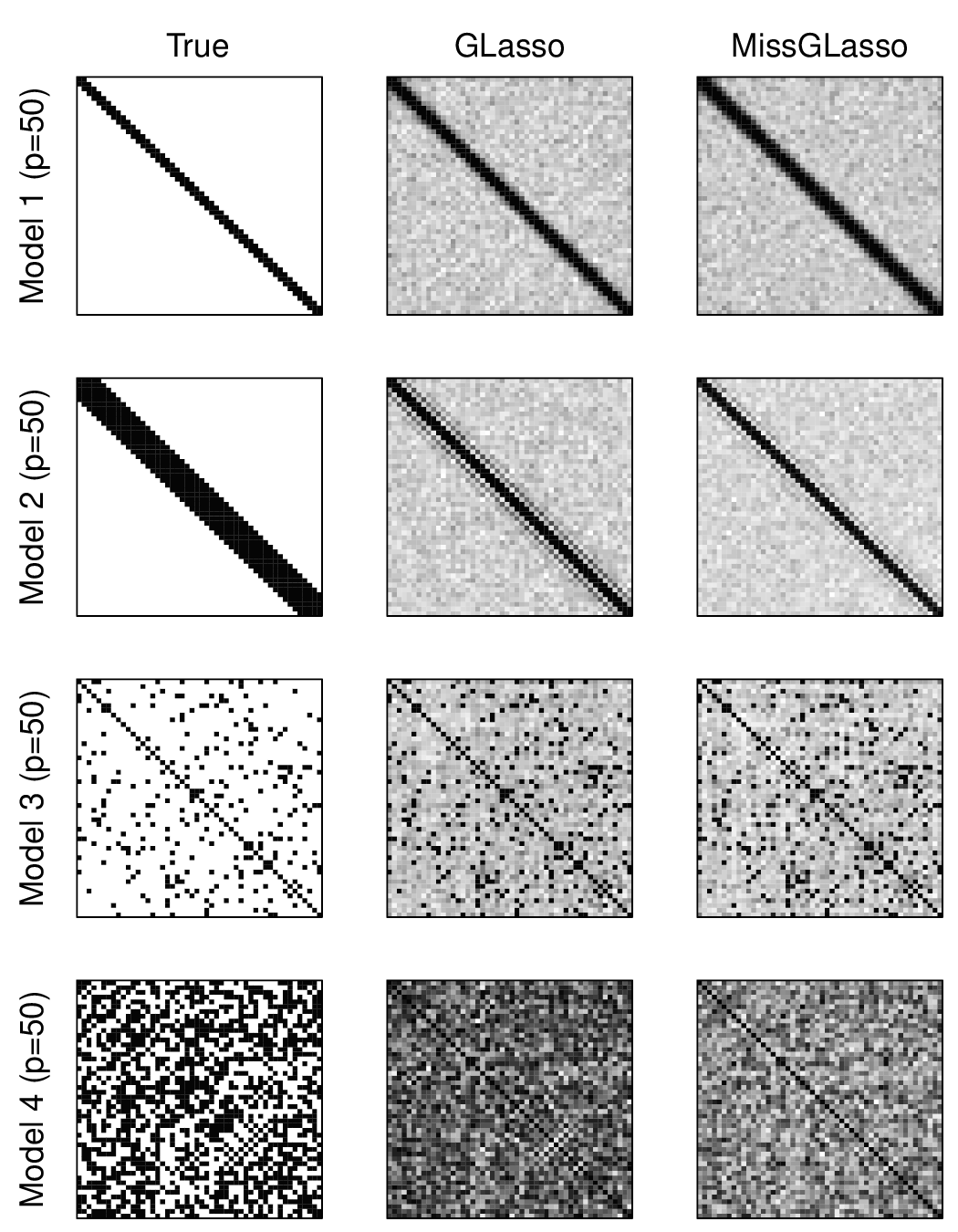} 
  \caption[]%<<-- Legende für 'List of Figures'  (ohne ``[..]'': 2x dasselbe).
  {Heat-maps of the identified zeros in the concentration matrix $K$ among 50
    simulation runs of models 1-4 with $p=50$. White color stands for zero in each of the 50 simulation
    runs. Black stands for non-zero in all runs. Left column: True
    concentration matrix. Middle column: Concentration matrix from GLasso
    applied on complete data. Right column: Concentration matrix from
    \emph{MissGLasso} applied on data with $30\%$ of the values missing}% eigentliche Bild-Legende.
\label{fig:heatmap}
%\end{centering}
\end{figure}

Finally, we comment on initialization and computational timings of the
\emph{MissGLasso}. In the above simulation we used the \emph{MeanImp}
solution as starting values $(\mu^{(0)},K^{(0)})$ for the \emph{MissGLasso}. For a
typical realization of model 2 with $p=100$, $30\%$ missing data and a prediction optimal tuned
parameter $\lambda$, our algorithm converges in $3.58$
seconds and $19$ EM-iterations. All computations were carried out with the
statistical computing language and environment \textbf{R} on a AMD
Phenom(tm) II X4 925 processor with 800 MHz cpu and 7.9 GB memory.

\subsubsection{Simulation 2: MissGLasso under MCAR, MAR and NMAR}\label{sec:simulation2}
In the simulation of Section~\ref{sec:simulation1} the missing values are
produced completely at random (MCAR), i.e., missingness does not depend on the
values of the data. As mentioned in Section~\ref{sec:missglasso} the
\emph{MissGLasso} is based on a weaker assumption, namely that the
data are missing at random (MAR), in the sense that the probability that a
value is missing may depend on the observed values but does not depend on
the missing values. A missing data mechanism where missingness depends also
on the missing values is called not missing at random (NMAR), see for
example \cite{LittleRubin}. In this
section we will show exemplarily that our method performs differently under
the MCAR, MAR and NMAR assumption.

We consider a Gaussian model with $p=30$, $n=100$ and with a block-diagonal
covariance matrix
\[
\Sigma=\begin{bmatrix}
  B & 0& \cdots&0\\
  0 & B & &\vdots\\
  \vdots& & \ddots&0\\
  0& \cdots & 0&B
\end{bmatrix},\qquad B=\biggl(\begin{smallmatrix} 1 & 0.7&0.7^{2} \\ 0.7 &
  1&0.7\\0.7^{2}&0.7&1 \end{smallmatrix}\biggl).
\]
Note that the concentration matrix $K$ is again block-diagonal and
therefore a sparse matrix.

We now delete values from the training data according to the following missing
data mechanisms:
\begin{itemize}
\item[1.]for all $b=1,\ldots,10$ and $i=1,\ldots,n$: 
\[\mathbf{x}_{i,3\cdot b}\quad \textrm{is missing if}\quad \eta_{i,b}=1,\]
where $\eta_{i,b}$ are i.i.d. Bernoulli random variables taking value $1$ with
probability $\pi$ and 0 with probability $1-\pi$.
\item[2.]for all $b=1,\ldots,10$ and $i=1,\ldots,n$: 
\[\mathbf{x}_{i,3\cdot b} \quad \textrm{is missing if}\quad
\mathbf{x}_{i,3\cdot b-2}<T.\]
\item[3.]for all $b=1,\ldots,10$ and $i=1,\ldots,n$: 
\[\mathbf{x}_{i,3\cdot b}\quad  \textrm{is missing if}\quad
\mathbf{x}_{i,3\cdot b}<T.\]
\end{itemize}

In all mechanisms the first and second variable of each block are
completely observed. Only the third variable of each block has missing
values. Mechanism 1 is clearly MCAR, mechanism 2 is MAR and mechanism 3 is NMAR. The probability
$\pi$ and the truncation constant $T$ determine the amount of missing values. In our
simulation we use three different degrees of missingness: (a) $\pi=0.25$,
$T=\Phi^{-1}(0.25)$, (b) $\pi=0.5$, $T=\Phi^{-1}(0.5)=0$ and (c)
$\pi=0.75$, $T=\Phi^{-1}(0.75)$. Here, $\Phi(\cdot)$ is the standard normal cumulative
distribution function. Setting (a) results in about $8\frac{1}{3}\%$, (b)
in $16\frac{2}{3}\%$ and (c) in $25\%$ missing data. In Figure~\ref{fig:simmar}, box-plots of
the Kullback-Leibler loss over 50 simulation runs are shown. As expected we see
that \emph{MissGLasso} performs worse in the NMAR case. This observation is
more pronounced for larger percentages of missing data.

\begin{figure}[!t]%--- Bild 'H'ier, 'B'ottom oder 'T'op  ``!'' ich WILL!!
\begin{centering}
 \includegraphics[scale=0.73]{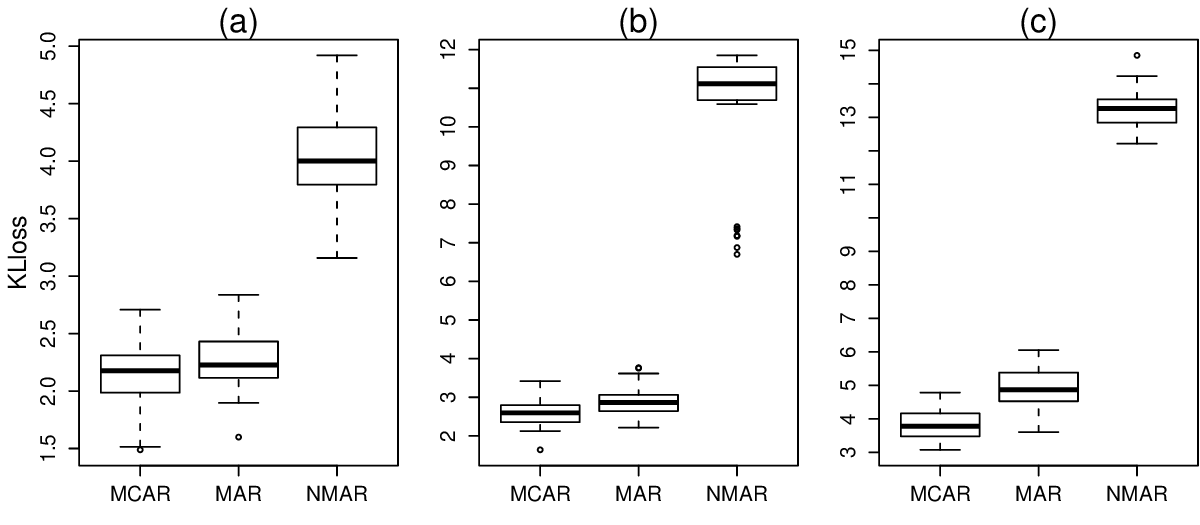} 
  \caption[]%<<-- Legende für 'List of Figures'  (ohne ``[..]'': 2x dasselbe).
  {Kullback-Leibler loss over 50 simulation runs for different missing data
    mechanisms (MCAR, MAR, NMAR) and
  different degrees of missingness: (a) $\pi=25\%$, $T=\Phi^{-1}(0.25)$, (b)
  $\pi=50\%$, $T=0$, (c) $\pi=75\%$, $T=\Phi^{-1}(0.75)$}% eigentliche Bild-Legende.
\label{fig:simmar}
\end{centering}
\end{figure}

\subsubsection{Simulation 3: BIC and cross-validation}
So far, we tuned the parameter $\lambda$ by minimizing twice the negative
log-likelihood (log-loss) on validation data. However, in practice, it
is more appropriate to use cross-validation or the BIC criterion
presented in Section~\ref{sec:tuning}.

Figure~\ref{fig:simtune} shows the Kullback-Leibler loss, the true positive
rate and the true negative rate for the \emph{MissGLasso} applied on model
1 with $p=50$. We see from the plots that cross-validation and tuning using
additional validation
data of size $100$ lead to very similar results. On the other hand BIC performs inferior
in terms of Kullback-Leibler loss, but slightly better regarding the true
negative rate.
\begin{figure}[!t!h]%--- Bild 'H'ier, 'B'ottom oder 'T'op  ``!'' ich WILL!!
\begin{centering}
 \includegraphics[scale=0.8]{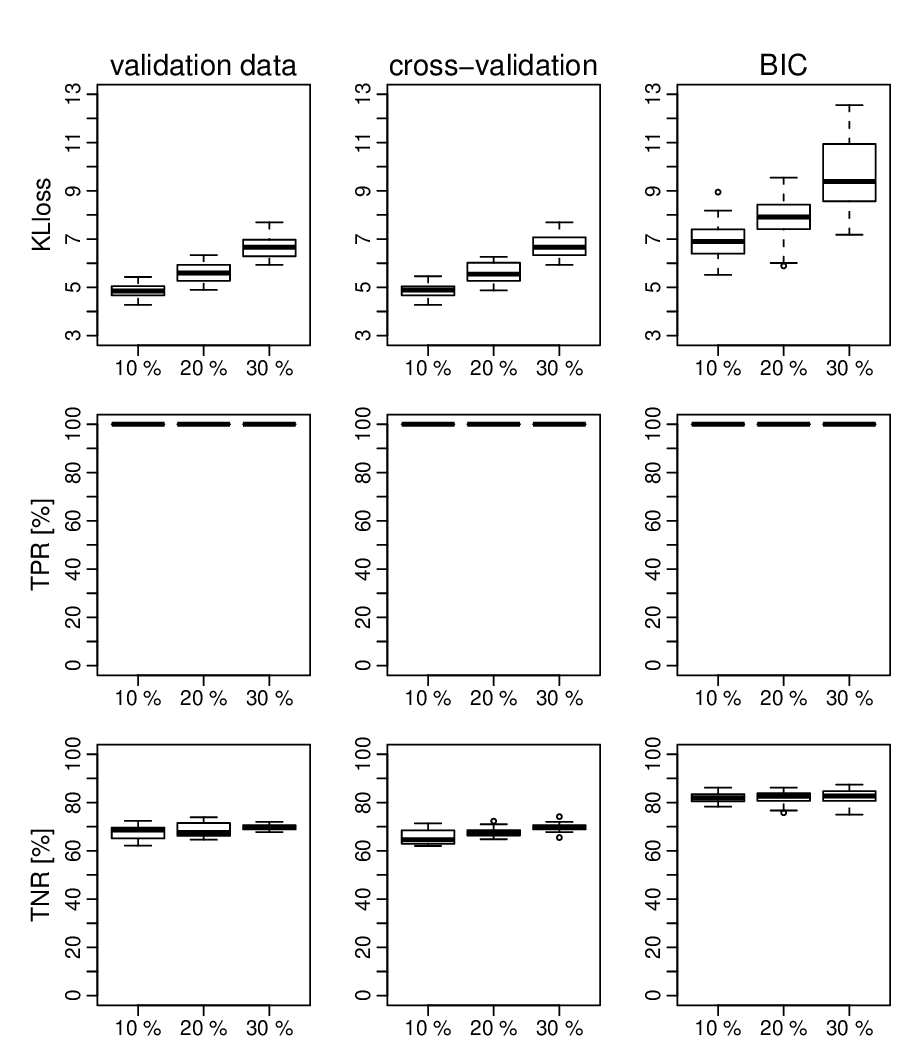} 
  \caption[]%<<-- Legende für 'List of Figures'  (ohne ``[..]'': 2x dasselbe).
  {KLloss, TPR, TNR of the \emph{MissGLasso} estimator tuned with either additional validation data, cross-validation or BIC. Model 1 with $p=50$, $n=100$
    and $10\%-30\%$ missing values, based on $50$ simulation runs}% eigentliche Bild-Legende.
\label{fig:simtune}
\end{centering}
\end{figure}
\subsubsection{Scenario 4: Isoprenoid gene network in Arabidopsis thaliana}\label{sec:isoprenoid}
For illustration, we apply our approach for modeling the isoprenoid gene
network in Arabidopsis thaliana. The number of genes in the network is
$p=39$. The number of observations, corresponding to different experimental
conditions, is $n=118$. More details about the data can be found in
\cite{wille}. The dataset is completely observed. Nevertheless, we produce
missing values completely at random and examine the performance of \emph{MissGLasso}. %with \emph{MeanImp} and \emph{MissRidge} and with
%K-nearest neighbors imputation (\emph{KnnImp}) introduced by
%\cite{knnimpute2001}. 
We consider the following experiments.

\paragraph{First experiment: Predictive performance in terms of log-loss.}
Besides \emph{MissGLasso}, \emph{MeanImp} and \emph{MissRidge} we consider here
a fourth method based on K-nearest neighbors imputation
\citep{knnimpute2001}. For the latter we impute the missing values by
K-nearest neighbors imputation and then we estimate the inverse covariance
by using GLasso on the imputed data. The number of nearest neighbors is chosen in
advance in order to obtain minimal imputation error.

Based on the original data we create 50 datasets by deleting (completely at
random) each time
$30\%$ of the values. For each of these datasets we compute a
10-fold cross-validation error as follows: We split the dataset into
10 equal-sized parts. We fit for various $\lambda$-values the different
estimators on every nine tenth of
the (incomplete) dataset and evaluate the prediction error (based on out-sample negative
log-likelihood) on the left-out part of the
original (complete) data. The cross-validation error (cv error) is then the average over
the 10 different prediction errors for an optimal $\lambda$-value. The
box-plots in the left panel of Figure~\ref{fig:realdata} show the cv errors
over the 50 datasets. \emph{MissGLasso}, \emph{MissRidge} and \emph{KnnImp} lead to a
significant gain in prediction accuracy over \emph{MeanImp}. In this example
\emph{MissRidge} performs best.   

\paragraph{Second experiment: Edge selection.}
First, we select using the GLasso on the original (complete) data (prediction optimal tuned) the twenty
most important edges according to the estimated partial
  correlations given by
\begin{equation*}
\hat{\rho}_{jj'|\mathrm{rest}}=\frac{|\hat{K}_{jj'}|}{\sqrt{\hat{K}_{jj}\hat{K}_{j'\!j'}}},\quad j,j'=1,\ldots,p.
\end{equation*}
Then, we create 50 datasets by producing completely at random $m\%$ missing
  values and select using the \emph{MissGLasso} for each of the 50 datasets the twenty most
  important edges according to the partial correlations
  $\hat{\rho}_{jj'|\mathrm{rest}}$. We do this for $m=5, 10, 15, 20, 25,
30$. 
Finally, we identify the overlap of the selected edges without missing values and of
  the selected edges with $m\%$ missing data. The box-plots in the right panel of Figure~\ref{fig:realdata}
  visualize the size of this overlap. Even with $30\%$ missing data, the
  \emph{MissGLasso} detects about 13 of the twenty most important edges of
  the complete data.

\begin{figure}[!t!h]%--- Bild 'H'ier, 'B'ottom oder 'T'op  ``!'' ich WILL!!
%\begin{centering}
  \includegraphics[scale=0.8]{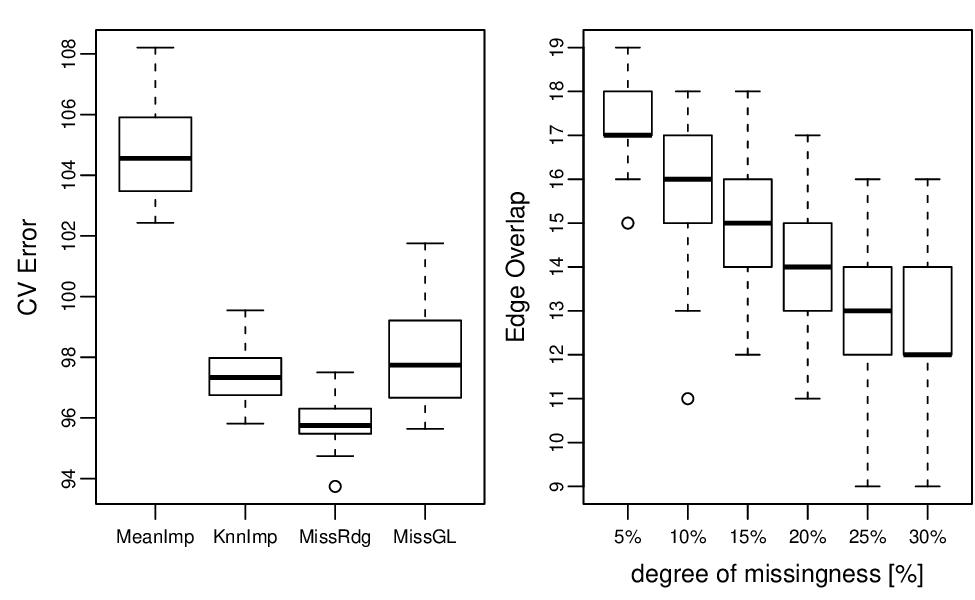} 
  \caption[]%<<-- Legende für 'List of Figures'  (ohne ``[..]'': 2x dasselbe).
  {Arabidopsis thaliana data ($n=118, p=39$). Left panel: Cross-validation error of \emph{MeanImp},
    \emph{KnnImp}(=K-nearest neighbors imputation followed by the GLasso),
    \emph{MissRidge} and \emph{MissGLasso} over $50$ datasets. For each
    dataset, $30\%$ of the original data are deleted. Right panel:
  Box-plots of the overlap of the twenty most important edges from GLasso
  and \emph{MissGLasso} with and without
  missing values over 50 datasets}% eigentliche Bild-Legende.
\label{fig:realdata}
%\end{centering}
\end{figure}

\subsection{Simulations for sparse regression}
\subsubsection{Simulation 1}
In this section we will explore the performance of the two-stage likelihood
method developed in Section~\ref{sec:two-stage}. In particular, we compare
our new method with alternative ways of treating high-dimensional regression with
missing values. 

Consider the Gaussian linear model
\begin{eqnarray*}
 &&Y_i=\beta^{T}X_i+\epsilon_i,\quad i=1,\ldots,n,\\
&&\epsilon_1,\ldots,\epsilon_n \quad\textrm{i.i.d.} \sim \calN(0,\sigma^{2}), \nonumber 
\end{eqnarray*}
where the covariates $X_{i}\in \R^{p}, i=1,\ldots,n$, are either fixed or
i.i.d. $\sim \calN(0,\Sigma)$. In all simulations training- and validation
data are generated from this model. Assuming that there are missing values
only in the $\mathbf{x}$ matrix of the training data we apply one of the
following methods:
\begin{itemize}
\item \emph{MeanImp:}
Impute the missing values by their corresponding column means. Then
apply the Lasso-estimator (\ref{eq:lasso}) on the imputed data.
\item \emph{KnnImp:} 
Impute the missing values by the K-nearest neighbors imputation
method \citep{knnimpute2001}. Then apply the Lasso on
the imputed data.
\item \emph{MissGLImp:}
Compute $(\hat{\mu},\hat{K})$ with the \emph{MissGLasso} estimator. Then,
use this estimate to impute the missing values by conditional mean
imputation, i.e., replace the missing values in observation $i$ by
$$\hat{x}_{\mathrm{mis},i}:=\E[x_{\mathrm{mis},i}|x_{\mathrm{obs},i},\hat{\mu},\hat{K}].$$ Finally, apply
the Lasso on the imputed data. 
\item \emph{Miss2stg:} This is the method introduced in
  Section~\ref{sec:two-stage}.
(\textbf{1st stage}: solve the \emph{MissGLasso} problem; \textbf{2nd
  stage}: estimate $\beta$ and $\sigma$ by minimizing a penalized negative
log-likelihood, see Equation (\ref{eq:2ndstage}), where we fixed $\mu$ and $K$ in the likelihood at the values
  from the 1st stage; initialization of EM with $\beta\equiv 0$ and
  $\sigma^{2}=\textrm{empirical variance of $\mathbf{y}$}$)
\end{itemize} 

All methods, except for \emph{MeanImp}, involve two tuning
parameters. Regarding the first parameter, the number of nearest neighbors in
\emph{KnnImp} or the regularization parameter for the \emph{MissGLasso} are chosen by
cross-validation on the training data. The second tuning parameter in the
Lasso or in the 2nd stage of the \emph{Miss2stg} approach, respectively, are chosen to minimize
the prediction error on the validation data.

To assess the performances of all methods we use the L2-distance between the
estimate $\hat{\beta}$ and the true parameter $\beta$,
$\|\hat{\beta}-\beta\|^{2}_2$.
\paragraph{First experiment:}
\begin{description}
\item[Model 5:]  $p=8$, $\Sigma_{jj'}=\tau^{|j-j'|}$ and
  $\beta$=(3,1.5,0,0,2,0,0,0).
\end{description}
We focus on four different versions of this model with different
combinations of $n/\tau/\sigma$, namely $20/0.5/3$;
$40/0.5/1$; $40/0.95/1$; $100/0.5/0.5$. The values $n/\tau/\sigma =
20/0.5/3$ correspond to the model which
was considered in the original Lasso paper \citep{tibshirani96regression}.

The box-plots in Figure~\ref{fig:mod0} of the L2-distances,
summarize the performance of the different methods for different
combinations $n/\tau/\sigma$. In this experiment, 20\% of the
training data were deleted completely at random. For reference, we added a
box-plot for the L2-distances for the Lasso carried out on complete data,
i.e., before deleting 20\% in the training data.

For the model from the original Lasso paper, namely the combination $n/\tau/\sigma=20/0.5/3$, we see
that the Lasso on complete data does not perform substantially better than
simple mean imputation on data with 20\% of the values removed. This is due to
the high noise level in this model. By increasing
$n$ and/or scaling down $\sigma$, we reduce the noise level and increase
the signal in the data. Indeed, in the setup
$n/\tau/\sigma=40/0.5/1$, the analysis with complete data performs now much
better than all analyses carried out on data with missing
values. We also see that the \emph{Miss2stg} method is slightly better
than the other methods. In the setup $n/\tau/\sigma=40/0.95/1$ we increase the correlation between the covariates by
setting $\tau$ from 0.5 to 0.95 and we notice that now \emph{KnnImp},
\emph{MissGLImp} and \emph{Miss2stg} outperform the ``naive''
\emph{MeanImp} which ignores the correlation among the different variables
in the imputation step. %Here the
%new method is as good as the complete analysis. 
Finally in the
last setup, $n/\tau/\sigma=100/0.5/0.5$, where $n$ is increased and $\sigma$ is reduced
again, the \emph{Miss2stg} method is much better than the other
methods. Thus, for the cases considered where missing data imply a clear
information loss (e.g., when the difference between complete and mean
imputed data is large), the new two-stage procedure is best. 

\begin{figure}[!t!h]%--- Bild 'H'ier, 'B'ottom oder 'T'op  ``!'' ich WILL!!
%\begin{centering}
  \includegraphics[scale=0.62]{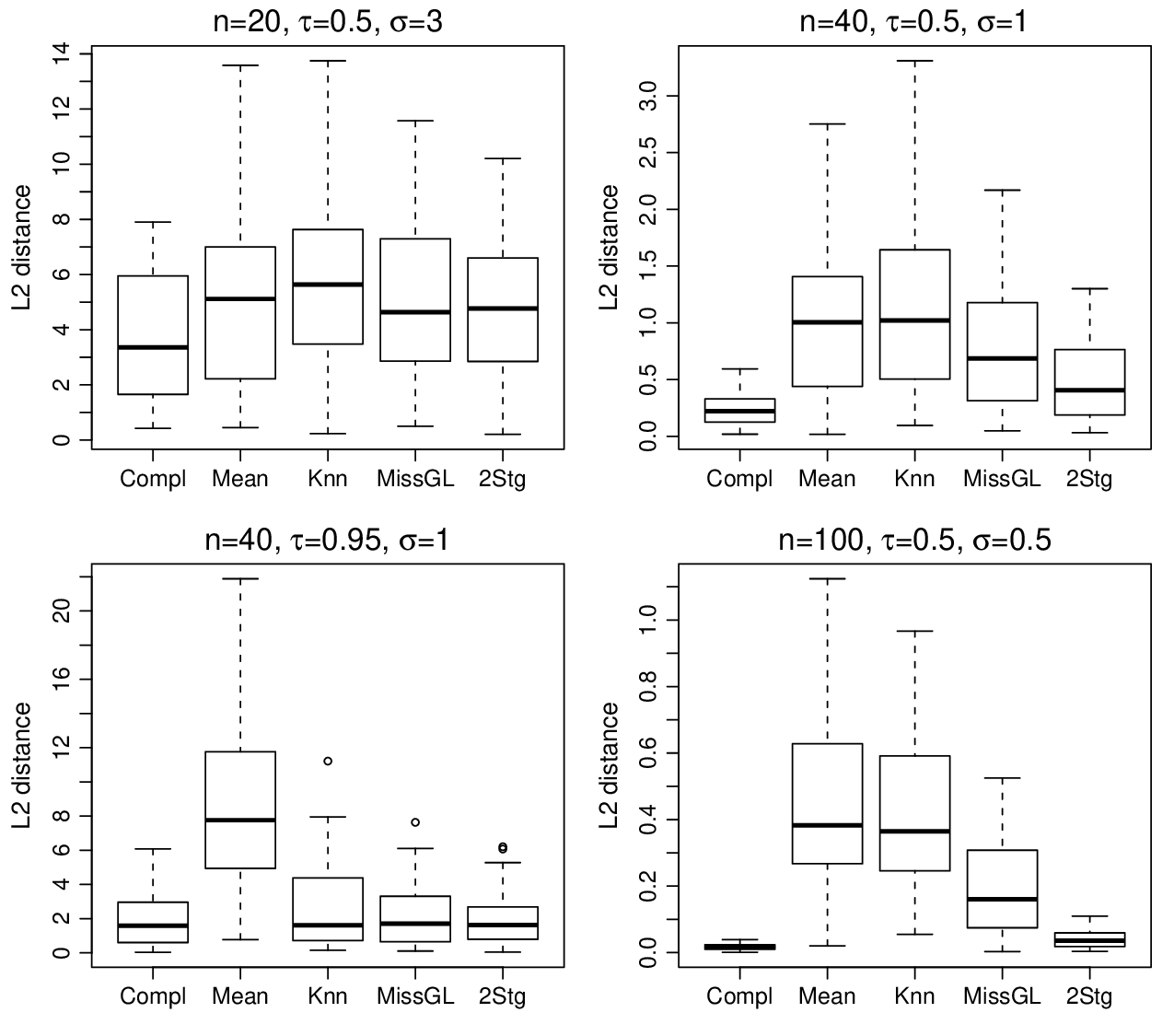} 
  \caption[]%<<-- Legende für 'List of Figures'  (ohne ``[..]'': 2x dasselbe).
  {Model 5. Box-plots of the L2-distances for different values for $n,\tau$ and
  $\sigma$ over 50 simulation runs with 20\% of the training data deleted
  completely at random. \emph{Compl}: Lasso on complete data (before deleting 20\%
  of the data). \emph{Mean(=MeanImp)}: Mean imputation followed by the
  Lasso. \emph{Knn(=KnnImp)}: Knn imputation followed by the
  Lasso. \emph{MissGL(=MissGLImp)}:
  MissGLasso and conditional mean imputation followed by the
  Lasso. \emph{2stg(=Miss2stg)}: Two-stage likelihood approach introduced in Section~\ref{sec:two-stage}}% eigentliche Bild-Legende.
\label{fig:mod0}
%\end{centering}
\end{figure}

\paragraph{Second experiment:}
Consider the following models:
\begin{description}
%\item Model 1: $n=100$; $p=30$; $\Sigma_{j,j'}=0.8^{|j-j'|}$; $\beta$:
%  randomly choose three non-zero elements (with value 2) in a vector
%  $\beta$ of size 30; $\sigma=0.5$.
%\item Model 2: $n=300$; $p=40$; $\Sigma_{j,j'}=0.5$ for $j\neq j'$, and
%  $\Sigma_{j,j}=1$; $\beta_j=2$ for $j=11,\ldots,20,31,\ldots,40$ and
%  zero elsewhere; $\sigma=3$.
\item[Model 6:]  $n=100$; $p=50$ and $p=200$; $\Sigma_{jj'}=0.8\times\mathrm{I}_{(j,j'\leq 9)}$ for $j\neq j'$, and
  $\Sigma_{jj}=1$; $\beta_j=2$ for $j=1,\ldots,8$ and zero
  elsewhere; $\sigma=0.5$.\vspace{0.15cm}
\item[Model 7:] $n=100$; $p=50$ and $p=200$;
  $\Sigma_{jj'}=\mathrm{I}_{(j=j')}$; $\beta=(3,1.5,0,0,2,0,0,0,\ldots)$; $\sigma=0.5$.\vspace{0.15cm}
\item[Model 8:]  $n=118$; $p=39$; $\textbf{x}$: data from isoprenoid gene
network in Arabidopsis thaliana (see Section~\ref{sec:isoprenoid});
$\beta_j=2$ for $j=1,2,3$ and zero elsewhere; $\sigma=0.5$.
\end{description}  
We delete 10\%, 20\% and 30\% of the training data
completely at random. The results (L2-distances) are reported in
Table~\ref{tab:sim2stage}. We read off from this table, that the
\emph{Miss2stg} method performs best in all three models. We further notice that in
model 7, \emph{KnnImp} and \emph{MissGLImp} do not perform better than
simple \emph{MeanImp} whereas \emph{Miss2stg} works much better than all other methods. The
explanation is that \emph{KnnImp} and \emph{MissGLImp} use the information
present in the covariance matrix of $X$, which is the identity matrix for
model 7, for imputation. On the other hand, our two-stage likelihood approach involves the joint
distribution of $(Y,X)$ which seems to be the main reason for its better performance.

% \begin{table}[!h]%%%Neue Resultate eingefügt
% \begin{center}
% \begin{tabular}{|c||r|r|r|r|}
%   \hline
%  \textbf{Model 1}& MeanImp & KnnImp & MissGLImp & Miss2stg \\ 
%   \hline

%   $10\%$ & 2.59 (0.18)& 1.22 (0.12)& 0.42 (0.04)& 0.32 (0.02)\\ \hline
%   $20\%$ & 5.87 (0.56)& 2.88 (0.23)& 1.16 (0.11)& 0.96 (0.08)\\ \hline
%   $30\%$ & 7.05 (0.47)& 5.61 (0.45)& 2.03 (0.18)& 1.46 (0.10)\\\hline
%  \multicolumn{5}{c}{}\\
% \hline
% \textbf{Model 2}& MeanImp & KnnImp & MissGLImp &Miss2stg \\\hline 
% $10\%$& 1.59 (0.15)& 0.49 (0.06)& 0.29 (0.04)& 0.13 (0.02)\\ \hline
%  $20\%$& 3.04 (0.17)& 1.37 (0.13)& 0.66 (0.06)& 0.25 (0.03)\\ \hline
%   $30\%$& 4.29 (0.22)& 2.38 (0.15)& 1.30 (0.12)& 0.62
%   (0.06)\\\hline
% \end{tabular}

\begin{table}[!h]%%%Neue Resultate eingefügt
\tabcolsep=1.5pt
\begin{center}
\begin{tabular}{|cr||c|c|c|c|}
  \hline
 \multicolumn{2}{|c||}{\textbf{Model 6}}& MeanImp & KnnImp & MissGLImp & Miss2stg \\ 
  \hline

  $p=50$&$10\%$ & 2.59(0.18)& 1.22(0.12)& 0.42(0.04)&\textbf{ 0.32(0.02)}\\ \hline
  &$20\%$ & 5.87(0.56)& 2.88(0.23)& 1.16(0.11)& \textbf{0.96(0.08)}\\ \hline
  &$30\%$ & 7.05(0.47)& 5.61(0.45)& 2.03(0.18)&\textbf{ 1.46(0.10)}\\\hline
 $p=200$&$10\%$ & 2.55(0.23)& 2.22(0.20)& 0.49(0.04)&\textbf{ 0.48(0.04)}\\ \hline
  &$20\%$ & 5.44(0.44)& 5.16(0.42)& \textbf{1.20(0.10)}& 1.23(0.08)\\ \hline
  &$30\%$ & 8.10(0.65)& 7.63(0.59)& 2.00(0.18)& \textbf{1.67(0.11)}\\\hline
\multicolumn{5}{c}{}\\
\hline
\multicolumn{2}{|c||}{\textbf{Model 7}}& MeanImp & KnnImp & MissGLImp & Miss2stg \\ 
  \hline

  $p=50$&$10\%$& 0.22(0.02) & 0.25(0.02) & 0.22(0.02) & \textbf{0.05(0.00)} \\ \hline
  &$20\%$ & 0.56(0.05) & 0.63(0.06) & 0.56(0.05) & \textbf{0.09(0.01)} \\ \hline
  &$30\%$ & 0.77(0.05) & 0.92(0.06) & 0.80(0.05) & \textbf{0.13(0.01)} \\ \hline
 $p=200$&$10\%$& 0.41(0.04) & 0.41(0.03) & 0.43(0.04) & \textbf{0.09(0.01)} \\ \hline
  &$20\%$ & 0.80(0.06) & 0.81(0.06) & 0.86(0.07) & \textbf{0.15(0.02)} \\ \hline
  &$30\%$ & 1.38(0.10) & 1.42(0.10) & 1.44(0.11) & \textbf{0.57(0.08)} \\  \hline
\multicolumn{5}{c}{}\\
\hline
\multicolumn{2}{|c||}{\textbf{Model 8}}& MeanImp & KnnImp & MissGLImp &Miss2stg \\\hline 
\multicolumn{2}{|c||}{$10\%$}& 1.59(0.15)& 0.49(0.06)& 0.29(0.04)& \textbf{0.13(0.02)}\\ \hline
 \multicolumn{2}{|c||}{$20\%$}& 3.04(0.17)& 1.37(0.13)& 0.66(0.06)&\textbf{ 0.25(0.03)}\\ \hline
  \multicolumn{2}{|c||}{$30\%$}& 4.29(0.22)& 2.38(0.15)& 1.30(0.12)& \textbf{0.62(0.06)}\\\hline
\end{tabular}

\caption{Models 6-8: Average (SE) L2-distance of \emph{MeanImp}, \emph{KnnImp},
  \emph{MissGLImp} and \emph{Miss2stg} with different degrees of missingness}
\label{tab:sim2stage}
\end{center}
\end{table}

\subsubsection{Scenario 2: Riboflavin production in Bacillus Subtilis}
We finally illustrate the proposed two-stage likelihood approach
on a real dataset of riboflavin (vitamin B$_2$) production by
\emph{Bacillus Subtilis}. The data has been provided by DSM
(Switzerland). The real-valued response variable is the logarithm of the
riboflavin production rate. There are $p=4088$ covariates (genes) measuring the
logarithm of the expression level of $4088$ genes and measurements of
$n=146$ genetically engineered mutants of Bacillus Subtilis. We compare the
estimators \emph{MeanImp}, \emph{KnnImp}, \emph{MissGLImp} and
\emph{Miss2stg} by carrying out a
cross-validation analysis as in the first experiment of
Section~\ref{sec:isoprenoid}. Here, we use the squared error loss
$(y-\beta^{T}x)^2$ to evaluate the prediction errors. To keep the
computational effort reasonable, we use only the $100$ covariates (genes)
exhibiting the highest empirical variances. The cv errors over
$50$ datasets (for each dataset, $30\%$ of the complete gene
expression matrix are deleted completely at random) are shown in
Figure~\ref{fig:riboflavin}. \emph{MeanImp} is worst. Our \emph{Miss2stg} performs
slightly better than \emph{KnnImp} and \emph{MissGLImp}.

\begin{figure}[!h]%--- Bild 'H'ier, 'B'ottom oder 'T'op  ``!'' ich WILL!!
\begin{centering}
  \includegraphics[scale=0.8]{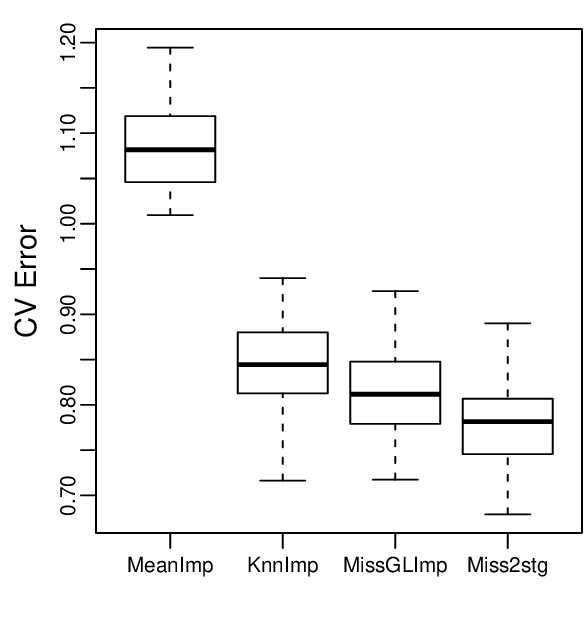} 
  \caption[]%<<-- Legende für 'List of Figures'  (ohne ``[..]'': 2x dasselbe).
  {Cross-validated prediction error $(y-\beta^{T}x)^2$ of \emph{MeanImp}, \emph{KnnImp},
    \emph{MissGLImp} and \emph{Miss2stg} over 50 datasets, where for each
    dataset 30\% of the riboflavin data are deleted
}% eigentliche Bild-Legende.
\label{fig:riboflavin}
\end{centering}
\end{figure}

\section{Discussion}
We presented an $\ell_{1}$-penalized (negative) log-likelihood method for estimating the
inverse covariance matrix in the multivariate normal model in presence of
missing data. Our method is based on the observed likelihood and therefore
works in the missing at random (MAR) setup which is more general than the
missing completely at random (MCAR) framework. As argued in
Section~\ref{sec:simulation2}, the method cannot handle missingness pattern
which are not at random (NMAR), i.e., "systematic" missingness.
For optimization, we use a simple and efficient EM algorithm which works in
a high-dimensional setup and which can cope with high degrees of missing
values. In sparse settings, the method works substantially better than $\ell_2$-regularization. In Section~\ref{sec:sparseregr}, the methodology was extended for
high-di\-men\-sion\-al regression with missing values in the covariates. We
developed a two-stage likelihood approach which was found to be never worse
but sometimes much better than K-nearest neighbors or using the
straightforward imputation with a penalized
covariance (and mean) estimate from incomplete data.

\begin{appendix}
\section{Proofs} 
\begin{proof}[Proposition \ref{prop:convergence}]
Denote by $f_{c}(\mathbf{x}|\mu,K)$ the multivariate Gaussian density of the
complete data. $f_{\mathrm{obs}}(\mathbf{x}_{\mathrm{obs}}|\mu,K)$ the density of the observed data. Furthermore, the conditional density of
the complete data given the observed data is $$k(\mathbf{x}|\mathbf{x}_{\mathrm{obs}},\mu,K)=f_{c}(\mathbf{x}|\mu,K)/f_{\mathrm{obs}}(\mathbf{x}_{\mathrm{obs}}|\mu,K).$$ The penalized
observed log-likelihood (\ref{eq:plik2}) fulfills the equation
\begin{eqnarray}\label{eq:eq1}
-\ell_{\mathrm{pen}}(\mu,K)&=&-\log
  f_{\mathrm{obs}}(\mathbf{x}_{\mathrm{obs}}|\mu,K)+\lambda\|K\|_{1}\nonumber\\
&=&\mathop{Q}(\mu,K|\mu',K')-\mathop{H}(\mu,K|\mu',K'),
\end{eqnarray}
where
\begin{eqnarray*}
\mathop{Q}(\mu,K|\mu',K')&=&-\E[\ell(\mu,K;\mathbf{x})|\mathbf{x}_{\mathrm{obs}},\mu',K']+\lambda\|K\|_{1}\\
\mathop{H}(\mu,K|\mu',K')&=&-\E[\log k(\mathbf{x}|\mathbf{x}_{\mathrm{obs}},\mu,K)|\mathbf{x}_{\mathrm{obs}},\mu',K'].
\end{eqnarray*}
By Jensen's inequality we get the following important relationship:
\begin{eqnarray}\label{eq:eq2}
\mathop{H}(\mu,K|\mu',K')\geq\mathop{H}(\mu',K'|\mu',K'),
\end{eqnarray}
see also \cite{wu}. $\ell_{\mathrm{pen}}(\mu,K)$, $\mathop{Q}(\mu,K|\mu',K')$ and
$\mathop{H}(\mu,K|\mu',K')$ are all continuous functions in all
arguments. Further, $\mathop{H}(\mu,K|\mu',K')$ is differentiable as a
function of $(\mu,K)$. If we think of $\mathop{Q}(\mu,K|\mu',K')$ and
$\mathop{H}(\mu,K|\mu',K')$ as functions of $(\mu,K)$ we write also
$\mathop{Q}_{(\mu',K')}(\mu,K)$ and $\mathop{H}_{(\mu',K')}(\mu,K)$.

Let $\theta^{m}=(\mu^{(m)},K^{(m)})$ be the sequence generated by the
EM algorithm. We need to prove that for a converging subsequence
$\theta^{m_j} \rightarrow \bar{\theta}$ ($j\to\infty$) the directional
derivative 
$-\ell'_{\mathrm{pen}}(\bar{\theta};d)$ is bigger or equal to zero for all directions
$d$ (\cite{tseng}). Taking directional derivatives
of Equation (\ref{eq:eq1}) yields
\[-\ell'_{\mathrm{pen}}(\bar{\theta};d)=\mathop{Q'_{\bar{\theta}}}(\bar{\theta};d)-\langle\nabla\mathop{H_{\bar{\theta}}}(\bar{\theta}),d\rangle.\]
Note that $\nabla\mathop{H_{\bar{\theta}}}(\bar{\theta})=0$ as
$\mathop{H}_{\bar{\theta}}(x)$ is minimized for $x=\bar{\theta}$ (Equation~(\ref{eq:eq2})). Therefore, it remains to show that $\mathop{Q'_{\bar{\theta}}}(\bar{\theta};d)\geq0$.
From the descent property of the algorithm (Equation (\ref{eq:eq1}) and
(\ref{eq:eq2})) we have:
\begin{equation}\label{eq:eq3}
-\ell_{\mathrm{pen}}(\theta^{0})\!\geq -\ell_{\mathrm{pen}}(\theta^{1})\!\geq\cdots\geq \!-\ell_{\mathrm{pen}}(\theta^{m})\geq -\ell_{\mathrm{pen}}(\theta^{m+1}).
\end{equation}
Equation (\ref{eq:eq3}) and the converging subsequence imply that $\{\ell_{\mathrm{pen}}(\theta^{m});m=0,1,2,\ldots\}$ converges to
$\ell_{\mathrm{pen}}(\bar{\theta})$. Further we have :
\begin{eqnarray*}
0\leq
\mathop{Q_{\theta^{m}}}(\theta^{m})-\mathop{Q_{\theta^{m}}}(\theta^{m+1})&=&-\ell_{\mathrm{pen}}(\theta^{m})+\ell_{\mathrm{pen}}(\theta^{m+1})+\underbrace{\mathop{H_{\theta^{m}}}(\theta^{m})-\mathop{H_{\theta^{m}}}(\theta^{m+1})}_{\leq 0}\nonumber\\
&\leq
&\underbrace{-\ell_{\mathrm{pen}}(\theta^{m})+\ell_{\mathrm{pen}}(\theta^{m+1})}_{\xrightarrow{m\to
    \infty}-\ell_{\mathrm{pen}}(\bar{\theta})+\ell_{\mathrm{pen}}(\bar{\theta})=0}.
\end{eqnarray*}
The first inequality follows from the definition of the M-Step. We conclude
\begin{equation}\label{eq:eq4}
\mathop{Q_{\theta^{m}}}(\theta^{m})-\mathop{Q_{\theta^{m}}}(\theta^{m+1})\xrightarrow{m\to\infty}0.
\end{equation}
In each M-Step we minimize the function $\mathop{Q_{\theta^{m}}}(x)$ with
respect to $x$. Therefore we have:
\begin{equation}\label{eq:eq5}
\underbrace{\mathop{Q_{\theta^{m_{j}}}}(\theta^{m_{j}+1})-\mathop{Q_{\theta^{m_{j}}}}(\theta^{m_{j}})}_{\xrightarrow{j\to \infty}
  0
 \quad(\ref{eq:eq4})}+\underbrace{\mathop{Q_{\theta^{m_{j}}}}(\theta^{m_{j}})}_{\xrightarrow{j\to
   \infty}\mathop{Q_{\bar{\theta}}}(\bar{\theta})}\leq\underbrace{\mathop{Q_{\theta^{m_{j}}}}(x)}_{\xrightarrow{j\to
   \infty}\mathop{Q_{\bar{\theta}}}(x)}.
\end{equation}
Using continuity, Equation (\ref{eq:eq4}) and Equation (\ref{eq:eq5}) we get
\[\mathop{Q_{\bar{\theta}}}(\bar{\theta})\leq\mathop{Q_{\bar{\theta}}}(x)
\qquad\forall x\]
and therefore, we have proven that
$\mathop{Q'_{\bar{\theta}}}(\bar{\theta};d)\geq0$ for all directions~$d$.
\end{proof}

\begin{proof}[Proposition \ref{prop:convergence2}]
The result follows from Proposition 5.1 and Lemma 3.1 in \cite{tseng}.
\end{proof}

\begin{proof}[Lemma \ref{lemma:mvn}]
We have
\begin{eqnarray}\label{eq:joint1}
(\epsilon_i,X_i)\!\sim\!\calN\left(\!(0,\mu),\left(\!\begin{array}{cc}\sigma^{2}&0\\
0&\Sigma\end{array}\!\right)\!\right)\,\textrm{and}\,\left(\!\begin{array}{c}Y_i\\X_i\end{array}\!\right)\!=\!\left(\begin{array}{cc}1&\beta^{T}\\
0&1\end{array}\right)\!\left(\!\begin{array}{c}\epsilon_i\\X_i\end{array}\!\right).
\end{eqnarray}
From (\ref{eq:joint1}) we see that the joint
distribution of $(Y_i,X_i)$ follows a (p+1)-variate normal distribution
with mean and covariance given by
\begin{eqnarray*}
\tilde{\mu}=(\beta^{T}\mu,\mu),\qquad \widetilde{\Sigma}=\left(\begin{array}{cc}\sigma^{2}\!+\!\beta^{T}\!\Sigma\beta&\beta^{T}\!\Sigma\\
\Sigma\beta&\Sigma\end{array}\right).
\end{eqnarray*}
The expression for the concentration matrix $\widetilde{K}=\widetilde{\Sigma}^{-1}$
can be derived by using the identity $\widetilde{\Sigma} \widetilde{K}=I$.
\end{proof}
\end{appendix}

\vspace{0.5cm}
{\bf Acknowledgements} {N.S. acknowledges financial support from Novartis International AG, Basel, Switzerland.} 

\bibliographystyle{spbasic}      % basic style, author-year citations
%\bibliography{diss240610}

\end{document}